\documentclass[journal]{IEEEtran}
\usepackage{amsmath,amsfonts}
\usepackage{algorithmic}
\usepackage{algorithm}
\usepackage{array}
\usepackage[caption=false,font=normalsize,labelfont=sf,textfont=sf]{subfig}
\usepackage{textcomp}
\usepackage{stfloats}
\usepackage{url}
\usepackage{verbatim}
\usepackage{graphicx}
\usepackage{cite}
\usepackage{xcolor}
\usepackage{amsthm}
\usepackage{mathrsfs}
\usepackage{booktabs}
\usepackage{hyperref}

\newcommand{\subsubsubsection}[1]{\noindent\textbf{•}\hspace{1ex}#1}

\newtheorem{definition}{Definition}

\newtheorem{lemma}{Lemma}
\newtheorem{theorem}{Theorem}

\begin{document}

\title{MLego: Interactive and Scalable  Topic Exploration Through Model Reuse}

\author{Fei Ye, Jiapan Liu, Yinan Jing, Zhenying He, Weirao Wang, X. Sean Wang

\thanks{F. Ye, J. Liu, W. Wang, Y. Jing, Z. He, X. Sean Wang are with the School of Computer Science, Fudan University, Shanghai 200438, China. Email: {fye21, 21210240251, 21212010036}@m.fudan.edu.cn, \{jingyn, zhenying, xywangCS\}@fudan.edu.cn.}

\thanks{Z. He is the corresponding author.}

}


\maketitle

\begin{abstract}

With massive texts on social media, users and analysts often rely on topic modeling techniques to quickly extract key themes and gain insights. 
Traditional topic modeling techniques, such as Latent Dirichlet Allocation (LDA), provide valuable insights but are computationally expensive, making them impractical for real-time data analysis. Although recent advances in distributed training and fast sampling methods have improved efficiency, real-time topic exploration remains a significant challenge. In this paper, we present MLego, an interactive query framework designed to support real-time topic modeling analysis by leveraging model materialization and reuse. Instead of retraining models from scratch, MLego efficiently merges materialized topic models to construct approximate results at interactive speeds. To further enhance efficiency, we introduce a hierarchical plan search strategy for single queries and an optimized query reordering technique for batch queries. We integrate MLego into a visual analytics prototype system, enabling users to explore large-scale textual datasets through interactive queries. Extensive experiments demonstrate that MLego significantly reduces computation costs while maintaining high-quality topic modeling results. MLego enhances existing visual analytics approaches, which primarily focus on user-driven topic modeling, by enabling real-time, query-driven exploration. This complements traditional methods and bridges the gap between scalable topic modeling and interactive data analysis.

\end{abstract}


\section{Introduction}
    
\IEEEPARstart{M}{assive} texts, such as tweets or reviews, are generated in social media applications like Twitter, Google Maps, and Airbnb. Most of these services provide users with the ability to access data based on user-specified attributes via location and time, such as "here and now". There is an increasing desire for users to extract insights from data via searching. For example, when users travel to unfamiliar places, they can plan their trips using Google Maps by analyzing comments and reviews written by local residents or tourists within the destination area to identify hot topics, events, etc.

Considering the example shown in Figure 1, Oliver is planning a trip to Paris. Besides popular attractions such as the Louvre Museum or the list of highly-rated spots provided by Google Maps, he wishes to discover niche tourist attractions to experience something off the beaten path by exploring people's reviews on Google Maps. As shown in Figure 1, Oliver can access reviews within different locations, such as $Q$. After analyzing reviews within $Q$, Oliver is particularly interested in some words: "Café des Deux Moulins, Amelie, Creme brulee." These words refer to Café des Deux Moulins, a famous location featured in the French film "Amelie", where the iconic scene involves the main character, Amelie, cracking the caramelized top layer of a creme brulee with a small spoon. Given Oliver's deep admiration for this film, he decides to include 'Café des Deux Moulins' in his travel itinerary, even though it's not on the top of the recommendation list.

To identify hot topics or events, topic modeling is an efficient way.
 We refer to Oliver's requirement as the problem of \textit{Analytic Query on Topic Modeling}. Given a user-specified region, Analytic Query on Topic Modeling models uncovers latent topics in the texts within this region. 




\begin{figure}[t]
\centerline{\includegraphics[width=0.43\textwidth]{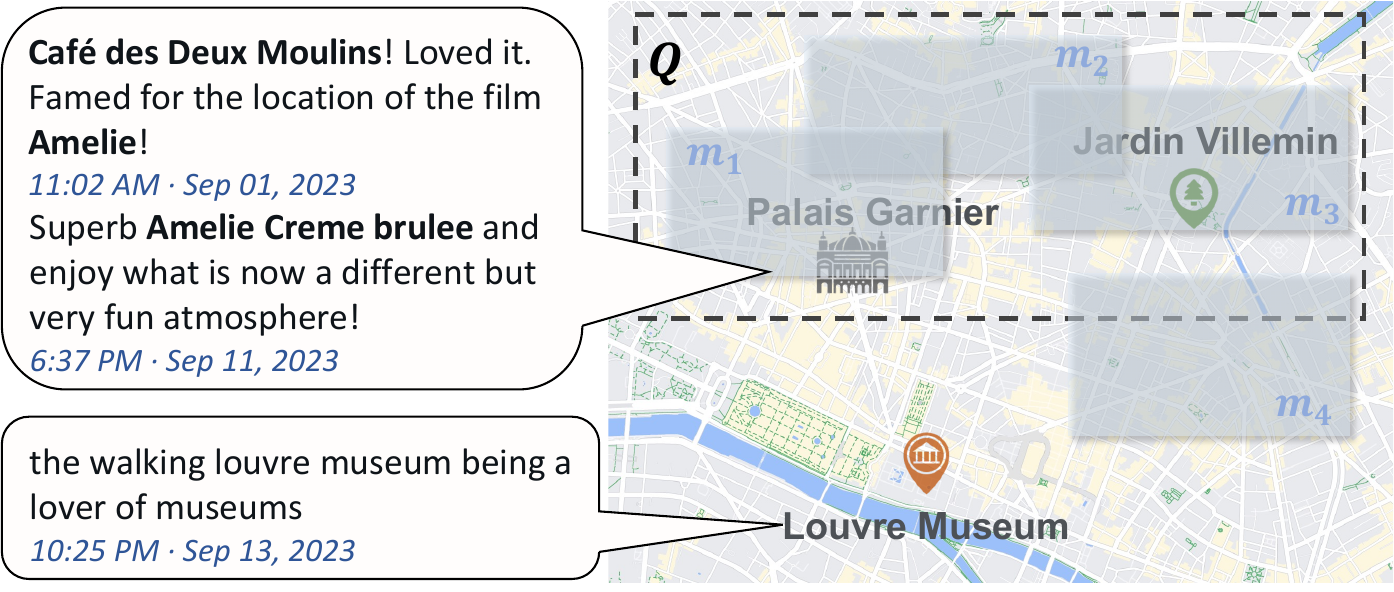}}
\caption{An example of an analytic query over a user-specified region.}
\label{fig_intro}
\end{figure}


Although existing topic modeling methods, such as LDA \cite{lda}, have proven effective in discovering latent topics within textual data \cite{LDA_survey_ACM, LDA_survey_IJACSA}, the challenge of prolonged model training time remains when handling large-scale datasets. For instance, in our scenario, when Oliver requests information for the $Q$ region, he repeatedly compares topics from reviews across different partitions to inform his travel planning. However, this process often requires frequent adjustments to ad-hoc analytics, with each modification involving a complete model training cycle. This makes the entire data analysis workflow cumbersome, posing significant challenges for supporting interactive data analysis tasks, such as Exploratory Data Analysis (EDA), which demand rapid response times. While several visual analytics systems\cite{Architext, topicLens, A23, A32B32, A9B8, B13, B14, B22, B44} have been developed for text data to facilitate human-in-the-loop topic modeling, these methods predominantly focus on visualizing already-trained models and steering the underlying topic modeling process. Achieving true interactive analysis, as described in our scenario, is still time-consuming and difficult to realize, unless a new model is constructed from scratch for each analysis, a process that remains highly inefficient and impractical. Despite recent efforts to accelerate model training through distributed techniques \cite{LDA*} and rapid sampling methods \cite{jmlr_online_17, AliasLDA, F+LDA, LightLDA}, it has been reported that commercial LDA applications still require the use of over 6,400 machines and several hours of computation time for model training \cite{LDA*, CuLDA}, severely restricting the practical applicability of LDA for large-scale real-time data analysis.


Inspired by commercial databases integrating machine learning (e.g., IBM System ML, Oracle ORE), we borrowed key concepts from data management—materialization and reuse to enhance visual analytics systems for interactive end-to-end analysis. Instead of discarding models after one-time use, we reuse materialized models to efficiently answer analytic queries, leveraging users’ tolerance for approximate models at interactive speeds. For instance, in exploratory analysis, data scientists often accept approximate ML models if they can rapidly obtain "close enough" estimates \cite{vldb_Hasani}. A key challenge is balancing model accuracy and response time for personalized customization. 
In Figure 1, to answer $Q$, the system can provide responses based on materialized models $m_1$, $m_2$, and $m_3$. Unlike stand-alone systems where users solely rely on their historical models, it is easy to see that materializing models for all user requests would result in a large number of models with overlapping predicate ranges. In our framework, when the model coverage within a query reaches a certain level, training time is significantly reduced. At this time, as there is an exponential number of candidate plans, enumerating all query plans would, in turn, become a hindrance to rapid data analysis.



To address these challenges, we propose "MLego", an analytic query framework that reuses materialized models like Lego bricks for interactive data analysis. MLego supports LDA model merging for two approximate posterior inference algorithms, allowing users to balance accuracy and efficiency via a user-defined cost function. It first searches for the optimal plan and then builds models online by reusing materialized ones. We further introduce a top-$k$-based query optimizer \cite{topk-survey} to accelerate plan search and optimize batch query execution.


In summary, we make the following contributions:

\begin{itemize}
\item We present an analytic query framework, "MLego," which efficiently executes LDA analytic queries through interactive model materialization and reuse,  accelerating model construction to enhance visual analytics systems for interactive end-to-end analysis.



\item We propose an optimization framework that combines hierarchical plan search for single queries with cost analysis and query reordering for batch queries, effectively reducing the overall execution cost by pruning the search space and optimizing query order.

\item We conduct extensive experiments on multiple datasets, demonstrating that MLego substantially accelerates the model construction process. Additionally, we implement a visualization prototype system for large-scale document analysis, utilizing the MLego framework to facilitate insightful data exploration.


\end{itemize}

\section{Related Work}
    
\textbf{Interactive Analytics for Topic Modeling}. In topic modeling, the primary goal is to facilitate interactive exploration of document collections, enabling users to extract insights and relationships through topic summaries of the entire corpus. Several approaches\cite{topicLens, A23, A32B32, A9B8} focus on interactive steering of the topic modeling process, allowing dynamic recomputation of topic models. For instance, TopicNets\cite{A23} iteratively updates topic modeling results based on user-selected document subsets. Additionally, human-in-the-loop topic modeling techniques allow users to steer underlying models to refine results\cite{B9, B11, B13, B14, B22, B44}. Our work, however, is orthogonal to these approaches. Unlike existing methods that focus on interactive modifications at the word, document, or topic level of a single model, we focus on enabling users to quickly construct topic models by merging existing LDA models via a query-driven approach, offering interactive analysis on large document corpora. Furthermore, while some works in Topic Modeling for Visual Analytics analyze combinations of different models for user-specified document subsets, we take a different approach by performing interactive analysis based on dynamic combinations of pre-trained topic models, rather than static visual analytics on a single model. In contrast to Mixed-Initiative systems, which require direct manipulation of model parameters and an in-depth understanding of the underlying model mechanisms\cite{B24, B16, B17, B47, B15, B26, B30, B38}, our approach presents users with semantically interpretable parameters. Users can simply specify a query range and desired model preferences to perform the interactive analysis.


\textbf{Model Training Acceleration}. Recently, there has been extensive work on speeding up model training, such as Gibbs samplers, distributed training methods, and online inference algorithms. Gibbs samplers aim to reduce the sampling complexity of LDA,  \cite{F+LDA, AliasLDA, LightLDA,fastLDA} are all designed for specific trade-off spaces, and \cite{LDA*} is designed as a hybrid sampler to choose a suitable sampler based on the length of documents. Distributed training methods\cite{F+LDA, YahooLDA, ELDA, AD-LDA, LDA*} partition documents and employ workers on them to scale up the training of LDA for large datasets. Research on online inference algorithms\cite{jmlr_online_17} has shown promising results in LDA training. Different from the above methods, MLego only necessitates the management of pre-trained models without expensive training resource costs, and our method can also be easily customized upon these systems.


\textbf{Analytic Query Speedup}. The study of speeding up analytic queries has gained interest in the database community recently. Techniques such as approximate query processing (AQP) and materialization are employed to support aggregate queries. AQP\cite{AQP}, along with similar techniques like sampling and corsets\cite{coreset1}, offers approximate results for tasks that require real-time performance. Materialization\cite{DataCube, materialization} involves partially computing results using OLAP cubes to answer queries. Our approach also relies on the fact that exact answers are not always required. We concentrate on analytic queries on topic modeling, and we treat models as intermediate results on datasets for materialization and reuse.

\section{Background}

\subsection{ML Primer}



\textbf{Latent Dirichlet Allocation (LDA)} \cite{lda, LDA_survey_ACM, LDA_survey_IJACSA} is a widely used hierarchical Bayesian probabilistic model of texts for topic modeling. Let $D$, $V$, $K$ be the number of documents, topics and unique words, respectively. Each topic defines a multinomial distribution over the vocabulary and is assumed to have been drawn from a Dirichlet with the parameter $\eta$, $
\beta_k = \left(\beta_{k v}\right)_{v=1}^V \sim \operatorname{Dirichlet}(\eta)$. Each document is an admixture of topics and the words in document $d$ are exchangeable. Each word $w_{dn}$ is from a latent topic $z_{dn}$ chosen in line with a document-specific distribution of topics $\theta_d=\left(\theta_{d k}\right)_{k=1}^K$ with the parameter $\alpha$. The posterior is

\begin{equation}
\begin{aligned}
&p(\beta, \theta, z| C, \eta, \alpha)
\propto 
\left[\prod_{k=1}^K \operatorname{Dirichlet}\left(\beta_k \mid \eta_k\right)\right] \\ &\cdot\left[\prod_{d=1}^D \operatorname{Dirichlet}\left(\theta_d \mid \alpha\right)\right]  \cdot\left[\prod_{d=1}^D \prod_{n=1}^{N_d} \theta_{d z_{d n}} \beta_{z_{d n}, w_{d n}}\right].
\end{aligned}
\end{equation}


For LDA and many other Bayesian models\cite{nips10_online}, computing the posterior is intractable. Two main approximate inference methods address this: sampling and optimization. Sampling methods, like Gibbs sampling, can handle arbitrary distributions and converge to exact inference\cite{jmlr_online_17}. Optimization methods, such as Variational Bayesian inference, approximate the posterior by optimizing a surrogate distribution.

\begin{table}[H]\hypertarget{notations}{}
  \caption{Frequently used notations in the paper}
  \label{tab:freq}
  \begin{tabular}{c|m{6.5cm}}
    \toprule
        \textbf{Notation}&\textbf{Meaning}\\
    \midrule
        $q$, $\alpha$ & A query, and a query-level weight parameter specified by a user.\parbox[c]{5cm}{}\\
        $p^*$ & A optimal query plan. \parbox[c]{5cm}{}\\
        $Q$ & A batch of query. \parbox[c]{5cm}{}\\
        $sc$ & The score of a plan $\mathcal{P}$ \parbox[c]{5cm}{}\\
        $m$, $m^*$ &  A model and the approximate model of plan $\mathcal{P}$ \parbox[c]{5cm}{}\\
        $\mathscr{F}$ & The ML operator performed by query $q$. \parbox[c]{5cm}{}\\
        $D$, $\sigma$ & A dataset, and the query predicate specified by a user. \parbox[c]{5cm}{}\\
        $M$ & The materialized model set. \parbox[c]{5cm}{}\\
        $M_i$ & The maximum number of iterations for LDA training. \parbox[c]{5cm}{}\\
        $\mathscr{A}$, $c_{t}$, $l_{p}$ & The performance metric, time cost.  \parbox[c]{5cm}{}\\
        $C$,$w_d$,$w_{dn}$ & A corpus, a document and a word. \parbox[c]{5cm}{}\\
    \bottomrule
\end{tabular}
\end{table}



\subsection{Query Model}

\textbf{Data, Pre-built Model:} 
For analytic queries, such as $Q$ in Figure 1, data is defined as $C=\{o_1, o_2, ... o_N\}$, where $N$ is the amount of data, and $o_i$ represents attributes like IP, timestamp, or tag of reviews. For model training, The data is defined as a corpus $C=w=\left(w_d\right)_{d=1}^D
$ of documents. Each document $w_d=\left(w_{d n}\right)_{n=1}^{N_d}$ is distributed iid conditioned on the global topics. The same definitions apply to training, validation, and test splits of data. The materialized model $m$ is defined as a tuple $<o, N, \Theta>$, where $N$ represents the model trained on $N$ data, $o$ is the dimension attributes, and $\Theta$ encapsulates parameters specific to the model type to support model merging.

\textbf{Analytic Query and Plan:} An analytic query $q$ is defined as an SQL statement with ML UDFs. In Figure 1, the query $Q$ searches for data via a certain IP range, and the output is the topic distribution generated by the LDA model. A plan $p$ is defined as a combination of materialized models, each with non-overlapping training data. The candidate plans for $Q$ are \{\{$m_1$\}, \{$m_2$\}, \{$m_3$\}, \{$m_1$, $m_2$\}\}.



\section{Problem Statement}

\subsection{System Architecture}

In Figure \hyperlink{overviewfig}{2}, the inputs to MLego are query predicates, datasets, model sets, and the weight parameter. MLego selects 
the model set and data set according to the query predicate, and then generates the optimal plan $p^*$ based on the weight parameter $\alpha$. We train according to $p^*$ to obtain the approximate model $m^*$.
MLego also provides a batch-query optimization strategy. When the input is the batch query, the relationship between multiple queries is calculated to reduce the overall return time.

\begin{figure}[htbp]\hypertarget{overviewfig}{} 
\centerline{\includegraphics[width=0.44\textwidth]{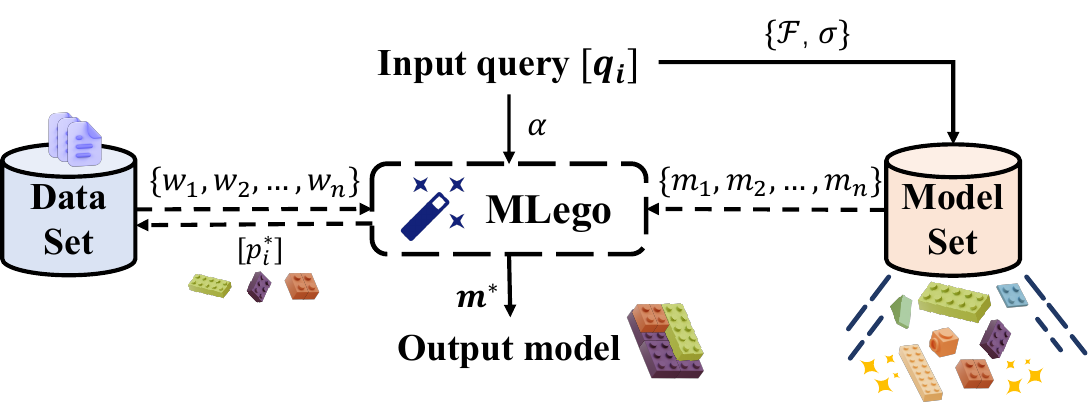}}
\caption{System architecture}
\label{fig_overview}
\end{figure}

\noindent\textbf{Definition 1} (Analytic Query on ML). An analytic query $q$ is characterized by a tuple

\begin{displaymath}
    \{\mathscr{F}, \alpha, D, \sigma, M\}
\end{displaymath}


where $\alpha$ is a weight parameter applying weights to different metrics on the query plan, and $D$ is the dataset used by $\sigma$ and $\alpha$ to produce the output model $m^*$. $M$ is the set of models obtained by applying the ML model $\mathscr{F}$ on $D$. In this paper, we focus on LDA, i.e., $\mathscr{F}$ is LDA.

\subsection{Plan Cost}\hypertarget{section3.2}{}

Plan Cost is the execution cost of the plan, including not only the time cost but also the performance loss due to model merging. Given a query plan, we will first \textbf{train} models using data $D$ uncovered by existing models, and then \textbf{merge} the trained models with pre-built models in this plan. Next, we will introduce what (1) time cost and (2) performance loss are in the above two processes, respectively.


\textbf{Performance Loss:} Performance loss measures the impact of model merging on performance. For a query $q$, let $m$ be the model built from scratch on $D_{q}$, and $m^*$ be the approximate model from merging. The performance loss for metric $\mathscr{A}$ is $l_{p} = \left|\mathscr{A}(m)-\mathscr{A}(m^*)\right|$.

\textbf{Time Cost:}
We use Time Cost to measure execution time of the plan $p^*$. We divide the time cost into training time and merging time cost. Among them, the training time is the cost we spend training the data uncovered by materialized models in $p^*$ and merging time is the time spent on model merging. We denote it as $c_{t}$.

It is worth noting that in this paper, we only divide plan cost into two types. In fact, plan cost may be more complicated, but this does not affect the effectiveness of MLego. Users can flexibly define more types of costs as needed.

\subsection{Problem Statement}\hypertarget{section3.3}{}
Based on the plan cost, we now define our score-based plan searching problem. We aim to find a query plan with both minimal time cost and optimal performance. We encode these two objectives in an intuitive linearized score function.

\begin{definition}
(Score Function). Let $\alpha$ $\in$ [0,1] be a weight parameter that represents the preference for performance in plan cost. Then, we define the score $sc$ of a plan $p$ as
\begin{equation}
sc = \alpha l_{p} + (1-\alpha) c_{t}.
\end{equation}


We linearize time cost and performance loss, weighting them via a user-defined parameter $\alpha$. A smaller $\alpha$ indicates stricter response time requirements, allowing users to intuitively control query return constraints. The optimal plan is determined by minimizing $sc$.

\end{definition}


\textbf{Plan Searching Problem:} The above score function allows finding the optimal query plan that shows a score less than any query plans with $sc(p) > 0$. The score-based plan searching problem is defined as follows:

\begin{definition}
(Score-based Plan Searching). Given the input query $q$, model set $M$, dataset $D$, as well as the weight parameter $\alpha$, find the optimal plan $p^*$ — from all possible plans $P$ — that satisfy the following condition:

\begin{equation}
\begin{aligned}
p^* = \mathop{\arg\min}_{p \in P} sc(p)
\enspace \text { s.t. } sc(p) > 0.
\end{aligned}
\end{equation}
\end{definition}

An intuitive method to solve this problem is to compute the score of all plans and then rank them. However, this method is time-consuming when $|M|$ is large. Thus, we propose a plan searching method in MLego. Table \hyperlink{notations}{1} summarizes the notations used in the paper.

\section{MLego Technique}
    \subsection{MLego: Model Merging Techniques}\hypertarget{section4}{}


\subsubsection{Model Merging for LDA}

Given a query $q$, our objective is to efficiently output the topic parameter $\widetilde{\beta}$ such that perplexity for $\widetilde{\beta}$ is close to the perplexity of $\beta$ where $\beta$ is the topic parameter obtained by running
LDA algorithm from scratch on the entire $C$. We seek to do this by only using $<o, N, \Theta>$ of model $m_i$. Different posterior approximation algorithms correspond to different model parameters $\Theta(m_i)$.


The following model merging methods, based on Mean-field Variational Bayes and Collapsed Gibbs Sampling, are all inspired by obtaining posterior distributions $p\left(\Theta \mid C_1, \ldots, C_b\right)$ via the recurrence relation:

\begin{equation}
\begin{aligned}
p\left(\Theta \mid C_1, \ldots, C_b\right) \propto p\left(C_b \mid \Theta\right) p\left(\Theta \mid C_1, \ldots, C_{b-1}\right),
\end{aligned}
\end{equation}

where $C_i$ is the batch of data, it automatically gives us the new posterior without needing to revisit old data points.

\begin{algorithm}
  \caption{Merging Bayesian Updating}\hypertarget{alg2}{}
  \begin{algorithmic}[1]
  \STATE \textbf{Input:} Set of LDA models $M$, $K$, $\eta$
  \STATE \textbf{Output:} $\theta$ of output model
  \STATE initialize $\lambda_{0}$ = $\lambda^{post}$ = $\eta$ ;\\
  \FOR{$m_i$ in $M$}
    \STATE ($N_i$, $\lambda_i$) in $\theta(m_i)$ ;\\
    \STATE $\Delta$$\lambda_i$ = $\lambda_i - \lambda_{0}$ ;\\
    \STATE $
    \lambda^{post} \leftarrow \lambda^{post}+\Delta \lambda
    $
  
  \ENDFOR
   \STATE  
    $
        q(\beta)=\prod_{k=1}^K \operatorname{Dirichlet}\left(\beta_k \mid \lambda_k^{post}\right)
    $
    \end{algorithmic}
\end{algorithm}

\subsubsubsection{\textbf{Mean-field Variational Bayes.}}
The idea of VB is to find the distribution $q_D$ that best approximates the true posterior $p_D$. We assume the approximating distribution, written $q_D$ for shorthand, takes the following form:

\begin{equation}
\begin{aligned}
&q_D(\beta, \theta, z \mid \lambda, \gamma, \phi)= \left[\prod_{k=1}^K q_D\left(\beta_k \mid \lambda_k\right)\right] \\
&\cdot\left[\prod_{d=1}^D q_D\left(\theta_d \mid \gamma_d\right)\right] \cdot\left[\prod_{d=1}^D \prod_{n=1}^{N_d} q_D\left(z_{d n} \mid \phi_{d w_{d n}}\right)\right].
\end{aligned}
\end{equation}

We have $q_D$ in the form of Eq.(5) and $p_D$ defined by Eq.(1). More specifically, the optimization problem of VB is defined as finding a $q_D$ to minimize the KL divergence between $q_D$ and $p_D$.



We use the variational parameters $\lambda$, $\gamma$, and $\phi$ to describe each topic, the topic proportions in each document, and the assignment of each word in each document to a topic, respectively. We assume the above three parameters are updated iteratively through the coordinate-descent algorithm. 

As the above VB-based posterior-approximation algorithm take Dirichlet distributions as prior for the topic parameters $\beta$ and return Dirichlet distributions for the approximate posterior of $\beta$, and the prior and approximate posterior are in the same exponential family, we employ "weighted" SDA-Bayes \cite{SGS} to achieve LDA merging. In a nutshell, the SDA-Bayes framework makes Streaming, Distributed, and Asynchronous Bayes updates to the estimated posterior according to a user-specified approximation batch primitive:

\begin{equation}
\begin{aligned}
p(\Theta \mid C_1, \ldots, C_B)
&\approx q(\Theta) 
\propto \left[\prod_{b=1}^B \mathcal{A}\left(C_b, p(\Theta)\right) p(\Theta)^{-1}\right] p(\Theta)\\
&\propto \exp \left\{\left[\xi_0+\sum_{b=1}^B\left(\xi_b-\xi_0\right)\right] \cdot T(\Theta)\right\}.
\end{aligned}
\end{equation}

$p(\Theta)$ is an exponential family distribution for $\Theta$
with sufficient statistic $T(\Theta)$ and natural parameter $\xi_0$, we have $\Theta = \beta$, $\mathcal{A} = BatchVB$, and $\xi = \lambda$. We merge models into a single model while taking into account their respective weights, which are determined based on the number of data points associated with each model. The pseudocode can be found in Algorithm \hyperlink{alg2}{1}. The time complexity is $O(n' \times K \times V)$ where $n'$ is the number of models to be merged.

\subsubsubsection{\textbf{Collapsed Gibbs Sampling.}}
Collapsed Gibbs Sampling(CGS) is also an applicable and very popular posterior approximation algorithm. Unlike variational methods that need to make unwarranted mean-field assumptions and require model-specific derivations, CGS can explore the sparse structure to asymptotically converge to the target posterior.

\begin{equation}
\begin{aligned}
p\left(z_{d i}=k \mid \boldsymbol{Z}^{-d i}, \boldsymbol{W}\right) \propto\left(N_{k d}^{-d i}+\alpha\right) \frac{N_{k v_{d i}}^{-d i}+\beta}{N_k^{-d i}+V \beta},
\end{aligned}
\end{equation}

where $N_{kd}$, $N_{kv}$ are topic-document count matrix and topic-word count matrix, which are sufficient statistics for the Dirichlet-Multinomial distribution.

We employ "weighted" Distributed Streaming Gibbs Sampling (DSGS)\cite{SGS} to achieve LDA merging. DSGS can perform CGS on different data partitions in a completely asynchronous fashion by fetching the global parameter $N_{kv}$, it can be viewed as a sequence of calls to the CGS procedure:

\begin{equation}
\begin{aligned}
\Delta \boldsymbol{N}_{k v}=C G S\left(\alpha, \beta+\boldsymbol{N}_{k v}, \boldsymbol{W}^t\right),
\end{aligned}
\end{equation}

where $\boldsymbol{W}^t$ is the document set corresponding to the $t$-th data batch. Given a unified $N_{kv}$, we regard the training process of each model as obtaining the update $\Delta N_{kv}$ based on Eq.(8). The process of model merging is to merge the updates:

\begin{equation}
\begin{aligned}
\boldsymbol{N}_{k v}=\lambda^m \boldsymbol{N}_{k v}^{t-1}+ \sum_{t=1}^{m} \lambda^{m-t}\Delta \boldsymbol{N}_{k v}^t,
\end{aligned}
\end{equation}

where $\lambda$ is the decay factor serving to weaken the posterior caused by merging and improve the performance of DSGS, each model only requires constant memory to store $\Delta N_{kv}$ and other metadata. We also consider the number of documents corresponding to each model while merging like VB. The pseudocode can be found in Algorithm \hyperlink{alg3}{2}. Since we calculate $\Delta N_{kv}$, the time complexity of model merging is the same as that of VB, which is $O(n' \times K \times V)$.

 \begin{algorithm}[H]
   \caption{Gibbs Sampling Updating}\hypertarget{alg3}{} 
   \begin{algorithmic}[1]
   \STATE \textbf{Input:} Set of LDA models $M$, $K$, $\beta_0$, decay factor $\lambda$
   \STATE \textbf{Output:} $\beta$ of output model
   \STATE initialize $N_{k v} = \beta_0 = 0$ \;
    \FOR{$m_i$ in $M$}
        \STATE ($N_i$,$\Delta N_{kv}^i$) in $\theta(m_i)$;\\
        \STATE $\beta$ = $\lambda$($\Delta N_{kv}^i$ + $N_{k v}$);
    \ENDFOR
    \STATE $
        \phi_{k v}=\frac{N_{k v}+\beta_0}{N_k+V \beta_0}
    $;\\
    \end{algorithmic}
  \end{algorithm}



It is easy to see that the above two model merging methods are model order-independent.





    \subsection{MLego: Single Query Optimization}

This section presents an efficient plan searching algorithm in MLego for reducing the plan searching costs, especially when model coverage in a query increases and plan searching becomes the main obstacle for fast data analysis.

\subsubsection{A Basic Approach and its Challenge}

One approach is to generate all candidate plans based on $M$ and $D$, followed by scoring and ranking. The main challenge of such a "generate-and-rank" method is that there are an exponential number of candidate plans, and the plan generation is time-consuming.


To solve this problem, we present Algorithm \hyperlink{alg1}{3}, which accelerates plan searching by introducing hierarchical plan generation to prune searching space. Next, we will present the details of the algorithm.

  \begin{algorithm}
   \caption{Single Query Optimization}\hypertarget{alg1}{}
   \begin{algorithmic}[1]
   \label{alg:MQP}
   \STATE \textbf{Input:} Query $q$, $\alpha$
   \STATE \textbf{Output:} model $m^*$
   
   \STATE $RL$ plans = $dfs(\Delta DAG)$\;
   \hypertarget{RLP}{}
   \IF{$\alpha$ = 1}
    \STATE $p^*$ = argmax($|M(p)|$) \;
  
  \ELSE
    \hypertarget{HPS1}{}
    \FOR{$p^*$ = threshold($L + \Delta L$)}
     \STATE Generate the next layer in each list;\\
     \STATE $\Delta L$ $\leftarrow$ $list_{l_p}$ if $\alpha$ $\neq$ 0 ;\\
     \IF{$max(|M(p_{RL})|) \leq x^*$}
     \hypertarget{opt1}{}
        \STATE $\Delta L$ $\leftarrow$ $list_{c_t(merge)+c_t(train)}$
     
     \hypertarget{opt2}{}
     \ELSE
        \STATE $\Delta L$ $\leftarrow$ $list_{c_t(merge)}$;\\
        \STATE $\Delta L$ $\leftarrow$ $RL$$(list_{c_t(train)})$;\\
     \ENDIF
  
    \ENDFOR
    \hypertarget{HPS2}{}
  \ENDIF
   \STATE Pick uncovered data in $p^*$ and train one model $m_u$; \\
   \STATE $m^*$ = Merge($m_u$, $m$ in $p^*$);\\
   \end{algorithmic}
  \end{algorithm}

\subsubsection{Cost Analysis}
Based on the plan cost in Section \hyperlink{section3.2}{4.B}, we further analyze each cost for two actions: (1) Merge and (2) Train. Merge conducts model merging, while Train handles online training for uncovered data.


\noindent\textbf{Performance Loss.} Since we train each materialized model completely, the performance loss is mainly caused by insufficient model retraining after merging.

We treat the cost model as an orthogonal issue that is often domain-specific. The only constraint that must be satisfied by the cost model is that it is monotonic. In other words, all things being equal, we assume that building a model merged by $x_i$ models should cost more than one with $x_j$ models if $x_i \geq x_j$. Our algorithm finds an optimal plan as long as the cost function is monotonic.

To verify the monotonicity assumption, based on recent research\cite{vldb_Hasani} on model merging, we conduct model merging experiments on some simple ML methods, as shown in Figure \hyperlink{fig2}{3}.  The experiment shows that the performance loss of the merged model increases as the number of merging increases compared with the model trained directly from scratch (i.e., \#models = 1). Each point in Figure \hyperlink{fig2}{3} represents the average value of all points within that interval. Taking the interval [301, 400] as an example, the above assumption is satisfied when the performance loss corresponding to the points within this interval remains unchanged or monotonically decreases.

\begin{figure}[htbp]
\centerline{\includegraphics[width=0.5\textwidth]{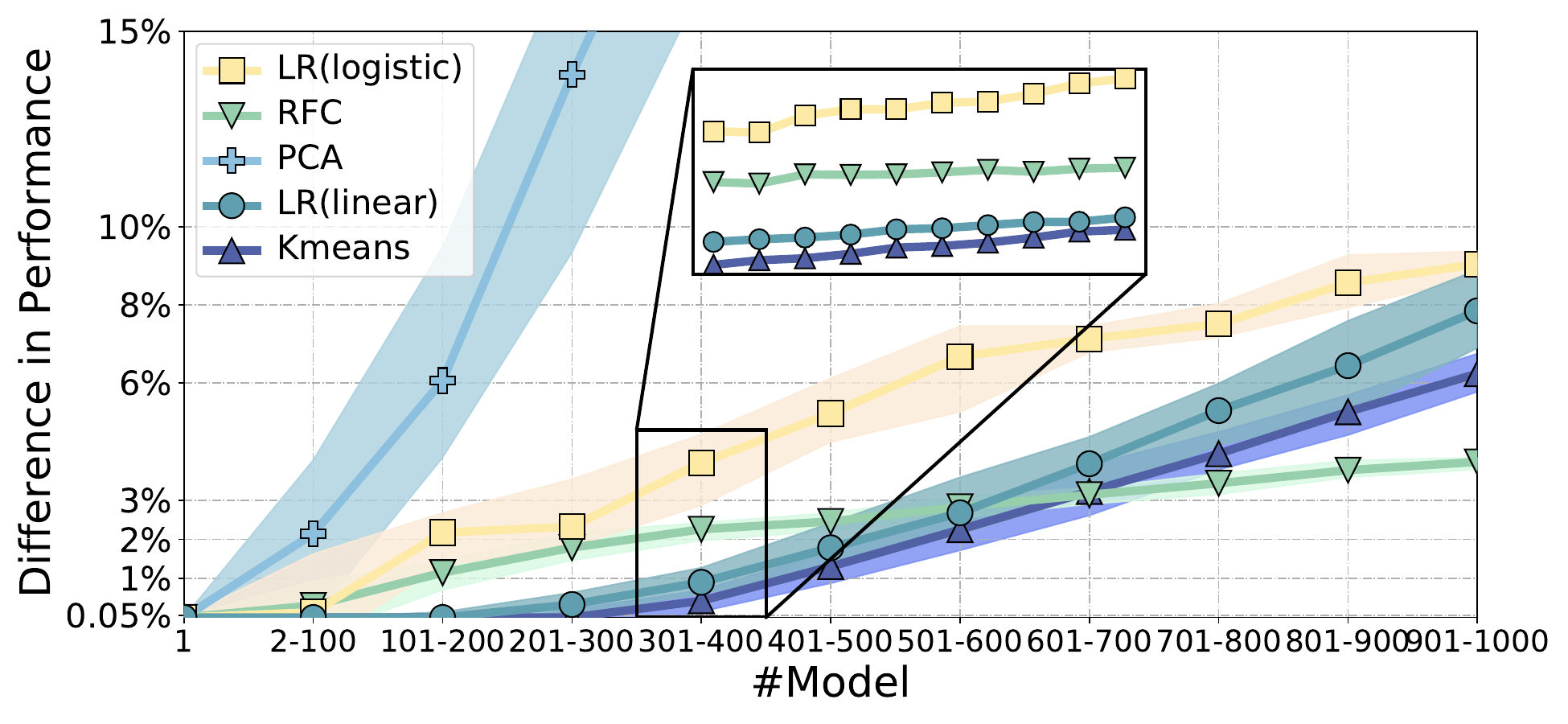}}
\caption{Verification for Monotonicity Assumption of Performance Loss}\hypertarget{fig2}{}
\label{fig_line11}
\end{figure}


Therefore, let $\mathcal{P}$(x) denote the monotone performance loss function where $x$ is the number of model merging and $\mathcal{P}(x) \in (0,1]$, i.e. $\mathcal{P}$($x_i$) $>$ $\mathcal{P}$($x_j$) iif $x_i > x_j$. As described in Section \hyperlink{section3.2}{4.B}, we denote $l_p$ as $\mathscr{A}(m) - \mathcal{P}(x)\mathscr{A}(m)$. Regardless of the full training effect of the model, that is, no matter how much $\mathscr{A}(m)$ is, we only care about the relative performance between $m$ and $m^*$, so $l_p = 1 - \mathcal{P}(x)$. When a query is covered by exactly one model, i.e. $x = 0$, $\mathcal{P}(x) = \mathcal{P}(0) = 1$, thus $l_p(Merge) = 1-\mathcal{P}(0) = 0$.





\noindent\textbf{Time Cost.} We use the algorithm complexity to represent time cost. As described in Section \hyperlink{section4}{5.A}, the complexity of training is $ O\left(M_{i}N^2 K\right)$ where $M_{i}$ is the maximum number of iterations, $N$ is the number of words and $K$ is the number of topics specified by users. The proof directly follows from \cite{lda}.

The complexity of LDA model merging is $ O\left(xKV\right)$, which is only related to the topic-word matrix and the number of merging, that is, only related to the number of topics $K$, the size of vocabulary $V$ and the number of materialized models $x$ in a plan.



\noindent\textbf{Takeaway.} The total performance loss is $1 - \mathcal{P}(x)$ - it grows to the number of merging $x$. The total time cost is $M_{i} N^2 K + xKV$ - it grows linearly to $M_{i}$, $x$, $K$ and $V$, quadratic to $N$. In general, $M_{i}$, $K$, and $V$ are set up by users. Thus, the total cost depends mostly on the number of merging and the amount of training data. Next, we will focus on reducing cost computing for speedup.


    

\subsubsection{Candidate Plans Generation}\hypertarget{section5.3}{}

We generate RL plans for the full plan list in order of training time cost for each query (line \hyperlink{RLP}{1}). For a fixed plan $p$, the training time cost is quadratic to the data uncovered by the models in $p$. As the data distribution in the query range is unknown, it is difficult to quickly estimate the training time cost. But when there is an inclusion relationship between the model sets in the two plans, we can easily compare their training time cost, i.e., for $\forall p_i$ and $p_j$, if $M(p_i) \subset M(p_j)$ where $M(p)$ is the model set of $p$, then $c_{t}(train)_{p_i} > c_{t}(train)_{p_j}$. We treat all plans with the model inclusion relationship as nodes in the same tree, and non-root nodes can be obtained by removing several models from the root nodes. We generate all relatively longest plans(RL plans) as the root nodes. That is, all other candidate plans can be generated from the RL plans, as shown in Theorem \hyperlink{theoremA2}{1}. All model non-overlapping relations in each query will be preserved in RL plans.

\begin{theorem}\hypertarget{theoremA2}{}
All possible candidate plans for each query $q$ can be generated by the RL plans.
\end{theorem}

\begin{proof}
    The proof of the Theorem 1 is obvious. We assume $p_i$ can not be generated by the RL plans, i.e., for $\forall p$ in RL plans, there is no way for $p_i$ to be obtained by removing several models in $p$. Let $p \cap p_i$ be the number of common models of $p$ and $p_i$, and we assume that $p$ is the plan that maximizes $p \cap p_i$. If $p_i$ cannot be obtained through RL plans, it means that the non-overlapping part of models between $p_i$ and $p$ has not been saved in RL plans, which conflicts with RL plans, so $p_i$ does not exist. 
\end{proof}

\subsubsection{Hierarchical Plan Searching}

We present a general plan searching algorithm based on the top-k algorithm (line \hyperlink{HPS1}{5} to \hyperlink{HPS2}{12}). Its main idea is to generate candidate plans hierarchically based on RL plans and score function, then utilize the top-k algorithm to avoid plan enumeration.

\noindent\textbf{Top-$k$ Algorithm.} As described in Section \hyperlink{section3.3}{4.C}, the score function consists of weighted $l_p$ and $c_t$. We use the threshold algorithm\cite{topk-survey,topk1} in the top-k algorithm family for optimization (line \hyperlink{HPS1}{5}). 


The threshold algorithm scans multiple ordered lists that represent different rankings of the same set of objects. The upper bound $th$ is computed and saved by applying a scoring function to partial scores of the last seen objects in different lists, where the last seen objects may be different, and the upper bound is updated every time a new object appears in one of the lists. Returns all the top k objects greater than or equal to $th$ by computing an overall score over the seen objects. 

\noindent \textbf{Example.} Consider the ranking lists given in Figure \hyperlink{fig4}{4(b)}. Since we look for the optimal plan with the minimal score, we compute the lower bound $th$ and return plans with scores lower than $th$. In the first layer $L_1$, no plan has a score lower than $th_{1_{min}}$. Then we scan plans in $L_2$. When we visit $p_2$, $sc(p_2) \leq th_{2_{min}}$, then $p_2$ returns, and we don't need to access plans in the remaining layers. 


\begin{figure}[htbp]
\centerline{\includegraphics[width=0.47\textwidth]{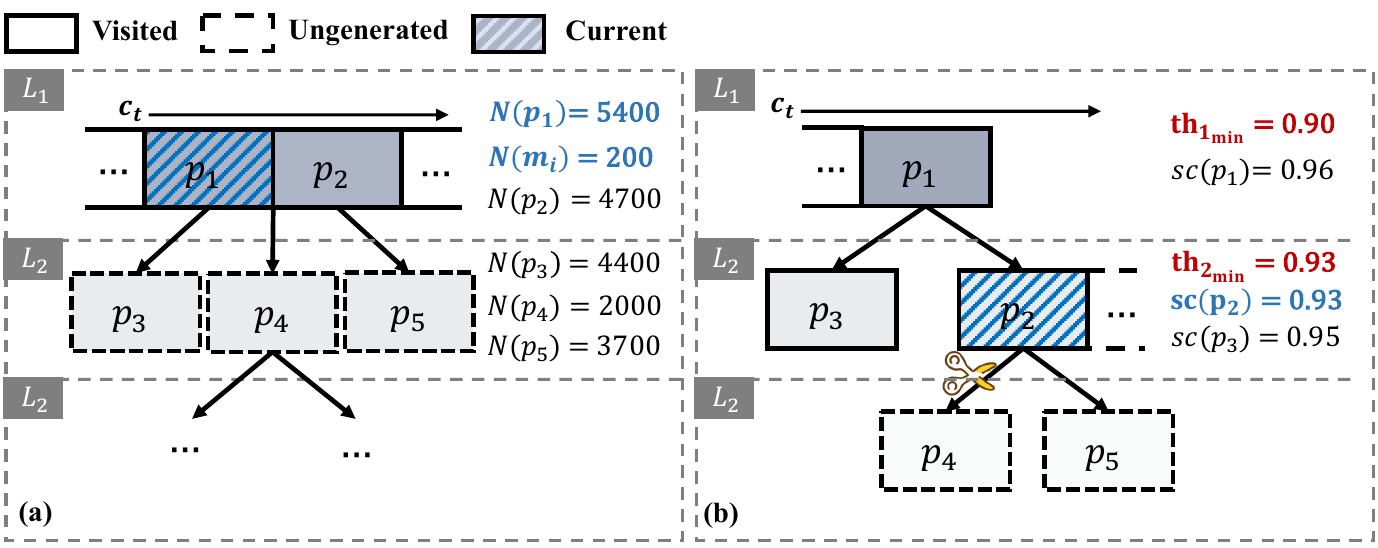}}
\caption{Single Query Optimization}\hypertarget{fig4}{}
\label{fig_psoa}
\end{figure}

Based on the threshold algorithm, we generate plans hierarchically across three lists ($l_p$, $c_t(merge)$, $c_t(train)$).

For lists of $l_p$ and $c_t(merge)$, since they are only positively correlated with $x$, we use the Breadth First Search(BFS) algorithm to generate plans hierarchically, with one more model per plan in $L_{i+1}$ than the number of models per plan in $L_i$. 


For the list of $c_t(train)$, we convert the training time cost computing based on data to be trained into that based on the number of models in the plan (line \hyperlink{HPS2}{12}). We regard RL plans as the candidate plans in $L_1$. When drilling down to the next layer $L_{i+1}$, we randomly remove a model in the plans from $L_i$ to obtain all candidate plans in $L_{i+1}$. 

The above method based on RL plans only guarantees that $c_t(train)_{p_i} \geq c_t(train)_{p_j}$ when $M(p_i) \subset M(p_j)$. To keep order within each list for the top-k algorithm, i.e., the $c_t(train)$ of all plans in $L_i$ is smaller than that in $L_{i+1}$, we sort the plan candidates generated by the $L_i$ and introduce the “push down” operation, as shown in Theorem 2. We sort plans in $L_i$ and push down plans that satisfy Theorem 2 to $L_{i+1}$, ensuring that the $c_t(train)$ list remains ordered. The essence of "push down" is to align the plan tree.


\begin{theorem}
Given $p_{i1}$ and $p_{i2}$ in $L_i$, if $N(p_{i1}) - N(min\ m\ in\ p_{i1}) \geq N(p_{i2})$, then $p_{i2}$ and plans generated based on $p_{i2}$ will be pushed down to $L_{i+1}$.
\end{theorem}

\begin{proof}

If $p_{i2}$ should be pushed down, we assume $p_{j}$ is generated by removing the minimum model $m$ on $p_{i1}$, then $N(p_{j}) > N(p_{i2})$, i.e., 
\begin{equation}
    \begin{aligned}
        N(p_{i1}) - N(min\ m\ in\ p_{i1}) = N(p_j) \geq N(p_{i2}).
    \end{aligned}
\end{equation}

As $p_{j}$ is in $L_{i+1}$ and $p_{i2}$ is in $L_i$, it is not guaranteed that $c_t(train)$ of all plans in $L_i$ is smaller than that in $L_{i+1}$. Therefore, maintaining the $c_t(train)$ list order is not guaranteed.

\end{proof}

\noindent \textbf{Example.} For example, in Figure \hyperlink{fig4}{4(a)}, $m_i \in M(p_1)$. Since $N({p_1}) = 5400 \ge N(p_1) - N(m_i) = 5200 \ge N(p_2) = 4700$, both $p_2$ and its child nodes ($p_4$ and $p_5$ in $L_2$) will be pushed down to $L_3$ as shown in Figure \hyperlink{fig4}{4(b)}.

\subsubsection{Improvement Computing Correlation of Costs}

The hierarchical plan searching discussed above involves multiple lists, which will weaken the searching efficiency if the generation directions of these lists are inconsistent. To further speed up plan searching, we consider reducing the number of lists (line \hyperlink{opt1}{8} to \hyperlink{opt1}{9}).

When the weight parameter $\alpha$ is 0, the score depends only on the time cost, i.e., $score = c_t(merge) + c_t(train)$, and plans in $c_t(merge)$ and $c_t(train)$ lists are generated in opposite directions. Considering the significant difference between merging and training time, we give a theoretical critical point $x^*$. If the number of models in the query plan is less than $x^*$, we can merge two lists. The intuition behind the optimization is that after we merge the $c_t(merge)$ and $c_t(train)$ lists, the merged list is in the same order as the original $c_t(merge)$ list, at least keeping the plan candidates of each layer unchanged.


\begin{theorem}\hypertarget{theorem5.2}{}
Given the maximum $|M(p)|$ among RL plans, $t_m$ is a single merging time cost, the merging time cost can be ignored without changing the order of plans in different layers of $c_t(train)$ if:

\begin{equation}
    \begin{aligned}
        |M(p)| \leq x^* =  \frac{c_t(min \, m\,in \,p_i)}{t_m}
    \end{aligned}
\end{equation}
\end{theorem}

\begin{proof}
For $\forall p_i \in set(p)$, $p_{ij}$ is generated through model filtering on $p_i$, i.e., $M(p_{ij}) \subset M(p_i)$. If 
\begin{equation}
    \begin{aligned}
        c_{t}(train)_{p_i} + c_{t}(merge)_{p_i} < c_{t}(train)_{p_{ij}} + c_{t}(merge)_{p_{ij}}
    \end{aligned}
\end{equation}
then,
\begin{equation}
    \begin{aligned}
        c_{t}(merge)_{p_i} - c_{t}(merge)_{p_{ij}} < c_{t}(train)_{p_{ij}} - c_{t}(train)_{p_{ij}}
    \end{aligned}
\end{equation}
The difference in the number of models between $p_i$ and $p_j$ is $\Delta m$, then 
\begin{equation}
    \begin{aligned}
        \Delta m * c_t(merge) < c_{t}(train)_{\Delta N(p_i - p_{ij})}
    \end{aligned}
\end{equation}
where $\Delta N(M(p_i) - M(p_{ij}))$ is the training data difference between two plans $p_i$ and $p_{ij}$, i.e., the training data corresponding to the model set that differs between $p_i$ and $p_{ij}$. Since $p_{ij} \subset p_i$, then
\begin{equation}
    \begin{aligned}
        \Delta m * c_t(merge) &< \Delta m * c_{t}(train)_{N(m_{min} in p_i)} \\
        &< c_{t}(train)_{\Delta N(p_i - p_{ij})}
    \end{aligned}
\end{equation}
Let $x^*$ be the maximum number of model merging and $c_t(min m in p_i)$ be the train time cost of the minimum model in plan $p$ such that
\begin{equation}
    \begin{aligned}
    x^* = \frac{ c_t(min\, m \,in \,p_i)}{t_m}
    \end{aligned}
\end{equation}
For $\forall p_i$ in this query, if the number of models of $p_i$ is less than $x^*$, then the merging cost can be ignored.

\end{proof}

Furthermore, we give a tighter upper bound $x^*$ in Theorem \hyperlink{theorem4}{4} so that ignoring the merge time cost will keep the order of plan candidates in each layer unchanged.

\begin{theorem}\hypertarget{theorem4}{}
Given the number of models $|M_p|$ included in the RL plan, $m_{min}$ is the model with the least training data in the model set. The merging time cost can be ignored if the time cost of a single merge $t_m$ meets:
\begin{equation}
\begin{aligned}
|M_p| \leq \frac{m_{min}}{t_m}
\end{aligned}
\end{equation}

\end{theorem}

\begin{proof}
    For $\forall p_i, p_j \in set(p)$. we assume that $c_t(train)_{p_i} < c_t(train)_{p_j}$. Let 
    \begin{equation}
    \begin{aligned}
        \Delta c_{t}^* &= min(c_t(train)_{p_i} - c_t(train)_{p_j}) \\
        &\leq c_t(train)_{m_{min}}
    \end{aligned}
    \end{equation}
    then 
    \begin{equation}
    \begin{aligned}
    c_t(train)_{p_i} + \Delta c_{t}^* < c_t(train)_{p_j}
    \end{aligned}
    \end{equation}
    then for $\forall p$, if $c_t(merge)_{p} < c_{t}^*$, then $c_t(merge)_{p}$ can be ignored, i.e., let $x^*$ be the maximum number of model merging such that $c_t(merge)_{p} = c_{t}^*$, i.e., 
    \begin{equation}
    \begin{aligned}
    x^* = \frac{\Delta c_{t}^*}{t_m}
    \end{aligned}
    \end{equation}
    For $\forall p \in$ RT plans, if the number of models of $p$ is less than $x^*$, then the merging cost can be ignored.
\end{proof}

    \subsection{MLego: Batch Query Optimization}

In this section, we study how to reorder query plans to minimize the total cost of batch queries in MLego. We first prove this problem is NP-hard and then present a heuristic algorithm \ref{alg:BQP}.


\subsubsection{Problem Analysis}


Given a batch of queries $Q$ = \{$q_1, q_2, ..., q_b$\}, each query has its corresponding weight parameter $\alpha$ and $\alpha$ is 0, the total return time for this batch is $T$. Given the current plan combination $P$ composed of the current plan of each query $q_i$, the execution time of each plan is $t_i$, we can reuse the overlapping ranges between query plans multiple times by training only once so that $\sum t_i \geq T$. Find the optimal plan combination $P^*$ that minimizes $T$.


As shown in Figure \hyperlink{fig5}{5(a)}, given a batch of queries $Q$ = \{$q_1, q_2, q_3$\}, we generate plans in $L_1$ for them based on Section 5. The current plan combination $P$ consists of the optimal plans $p^{q_1}_1$, $p^{q_2}_1$(\{$m_1$, $m_2$, $m_4$\}) and $p^{q_{3}}_1$(\{$m_1$, $m_2$, $m_3$\}) corresponding to each query. In the single query processing workflow, we need to train data uncovered by models and then return the results. When considering both $q_1$, $q_2$ and $q_3$ simultaneously, we can train the ranges $\Delta r_1$, $\Delta r_2$ and $\Delta r_3$ once, which can be used for the model merging of all queries, resulting in a time-saving of $c_t(\Delta r_1) + c_t(\Delta r_2) + c_t(\Delta r_3)$. In addition, although the training time for $p^{q_{3}}_2$ is longer than that for $p^{q_{3}}_1$, if the overlap of data between $p^{q_{3}}_2$ and $p^{q_{i}}_1$ leads to greater time savings, we consider \{$p^{q_{3}}_2$, $p^{q_{2}}_1$, $p^{q_{1}}_1$\} to be optimal to \{$p^{q_{3}}_1$, $p^{q_{2}}_1$, $p^{q_{1}}_1$\}. Finding an optimal plan combination $P^*$ is NP-hard, as shown in Theorem \hyperlink{theoremA4}{5}.

\begin{theorem}\hypertarget{theoremA4}{}
 The batch query optimization, which checks for every plan of each query to maximize the time savings, is NP-hard. 

\end{theorem}

\begin{proof}

We demonstrate the problem is NP-hard by reducing maximum coverage, which is a standard NP-complete problem, to the batch query optimization problem.

The maximum coverage problem is defined as, given a collection of sets $S = {S_1, S_2, ..., S_m}$ and a number $k$, is there a subset $S'$ of $S$, such that $\left|S^{\prime}\right| \leq k$ and number of covered elements $\left|\bigcup_{S_i \in S^{\prime}} S_i\right|$ is maximized.

The maximum coverage problem can be reduced to batch query optimization as follows. We treat each object $m$ in $S$ as a model $m$ in a query plan $p$. Then

\begin{itemize}
\item we set $k$ as the number of batch queries $Q$.
\item we partition the collection of sets into $|Q|$ parts, where each part corresponds to a query $q$, and we pick up $S'$ from $|Q|$ partitions respectively.
\item we replace the covered elements finding function $\left|\bigcup_{S_i \in S^{\prime}} S_i\right|$ with the time savings function, which takes batch query plans as input and outputs training time savings from batch query common range/region model reuse.
\end{itemize}

The Batch query optimization problem tells us if there is a set of query plans $P$ such that the total time savings is maximum. The answer to the maximum coverage problem is yes, if and only if the answer to the batch query optimization problem is yes.
\end{proof}








    

\subsubsection{Batch Query Optimization Algorithm}
We provide a heuristic algorithm based on Theorem \hyperlink{theoremA6}{6}. The key idea is to balance the benefit change and the training time cost change plan by plan for each query within the fixed query order in the batch. First, we define \textbf{Benefit} as follows.

  \begin{algorithm}[H]
   \caption{Batch Query Optimization Algorithm}
   \begin{algorithmic}[1]
   \label{alg:BQP}
    \STATE Filter $dag_i$ for $q_i$ in $Q_b$\;

    \FOR{$q_i$ in $Q_b$}
    
        \STATE $P^{q_i}_{L_1}$, $M_i$ = max(DFS($dag_i$))\;
        \FOR{$m_i$ in $M_i$}
        
            \STATE $\Delta \mathcal{B}_{m_i}$ = $c_{t}(\Delta(m_i, P_{-q_i}))$-$c_{t}(m_i)$
        
        \ENDFOR
        \FOR{$p_i$ in $P^{q_i}_{L_1}$}
        
            \FOR{$m_i$ in $M_i$}
            
                \STATE $p^*_i$ = $p_i$ - $m_i$ if $\Delta \mathcal{B}_{m_i} > 0$
            \ENDFOR
            \STATE $sc(p^*_i) = \sum \mathcal{B}_{m_i}$ for $m_i$ in $p^*_i$
            \STATE $sc(p^*_i) -= \Delta t(p_1, p_i)$
        \ENDFOR
        \STATE $P$ $\leftarrow$ $max(sc(p^*_i))$
    \ENDFOR
    \STATE \textbf{Return} $P$
    \end{algorithmic}
  \end{algorithm}

\begin{figure}[htbp]
\centerline{\includegraphics[width=0.5\textwidth]{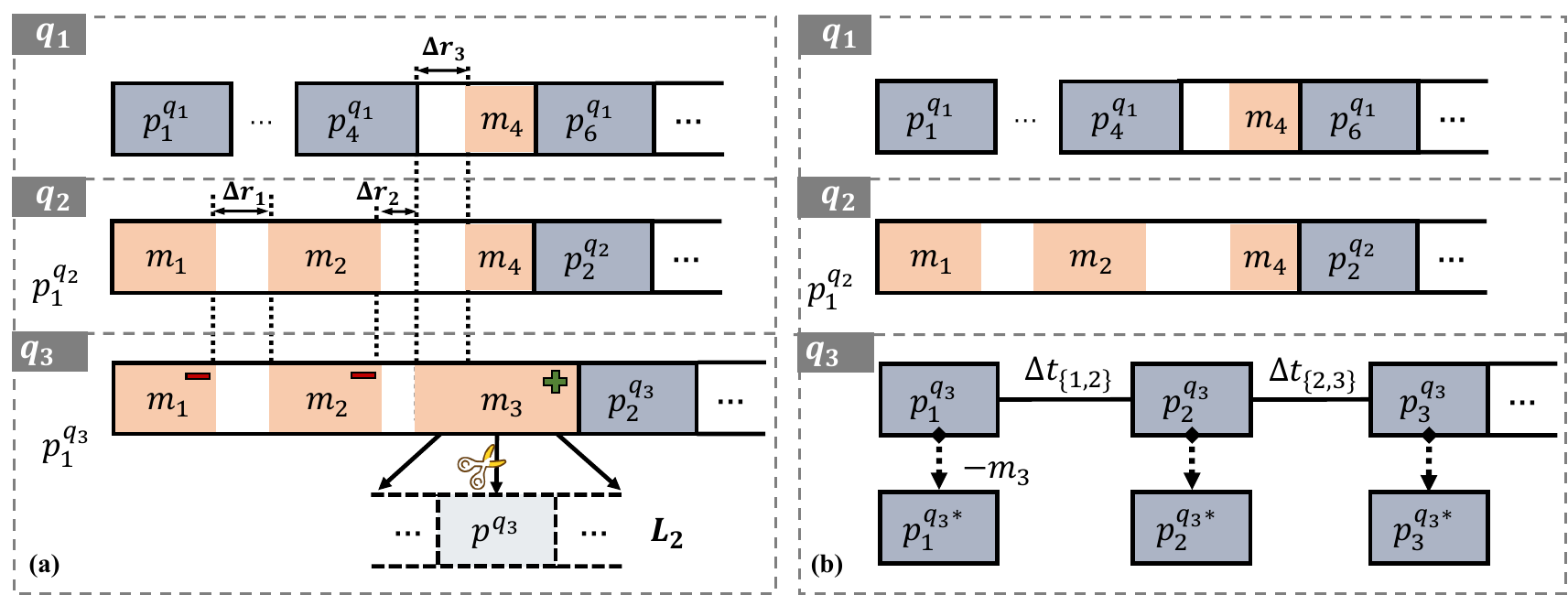}}
\caption{Batch Query Optimization}\hypertarget{fig5}{}
\label{fig_bqo}
\end{figure}

\begin{definition}
Given a batch of query $Q_b = \{q_1, q_2, ..., q_n\}$ and its current plan combination $P$, $\Delta(P) = \{\Delta r_1, \Delta r_2, ... , \Delta r_m\}$ represents the overlapping ranges of $P$, and the cost of model training on $\Delta(P)$ is $c_t(\Delta(P)) = \sum_{i=1}^{m} c_t(\Delta r_i)$, and the benefit $\mathcal{B}(P)$ is:

\begin{displaymath}
\sum_{i=1}^{m}(|\Delta(r_i)|-1)c_t(\Delta(r_i)),
\end{displaymath}
\end{definition}

where $|\Delta(r_i)|$ is the number of plans for $\Delta(r_i)$. For example, in Figure 5(a), the benefit of $P = \{p^{q_{3}}_1, p^{q_{2}}_1, p^{q_{1}}_1\}$ is $\mathcal{B}(P) = c_t(\Delta r_1) + c_t(\Delta r_1) + c_t(\Delta r_3)$.

We sequentially select plans from $q_1$ to $q_n$ based on the benefit, aiming to maximize the benefit of the current query plan combination $P$. Since our initial plan combination consists of the top-1 plan from each query, i.e., \{$p^{q_1}_1$, $p^{q_2}_1$, ..., $p^{q_n}_1$\}, while considering the maximum benefit, we also need to calculate the training time difference $\Delta t_{\{1,j\}}$ between the current plan $p^{q_i}_j$ and $p^{q_i}_1$. This refers to the time cost difference caused by the difference in the data to be trained for $p^{q_i}_1$ and $p^{q_i}_j$.


        
    


Based on Section 5.B, an approach is to generate candidate plans by layers and terminate the iteration based on Theorem \hyperlink{theoremA6}{6}. However, this method requires computing the benefit $\mathcal{B}(P)$ for each plan combination, which can be computationally costly. 

To solve the above problem, we compute the benefit of each model within query $q_i$ under the current plan combination of other queries $P^{-q_i}$. This involves computing the benefit of removing $m$ from the range of $q_i$, along with the training time cost of $m$. Let $\mathcal{B}(\{m, P^{-q_i}\})$ represent the benefit of $m$, where \{$m, P^{-q_i}$\} is the plan combination and $q_i$ has been replaced with one query with a range corresponding to $m$ but without including any models.

\begin{lemma}
    For $\forall p_x$ in $L_x$, $\exists P = \{p_1, p_2, ..., p_{x-j+1}\}$ in $L_j$, $p_x$ can be generated by removing $|x-j|$ models from $p$, where $p \in P$.
\end{lemma}

\begin{proof}
    According to Theorem \hyperlink{theoremA2}{1}, for $\forall p_x$ in $L_x$, $\exists p_1$ in $L_1$, $p_x$ is generated by removing ($x-1$) models. Let $M_{x-1}$ be these $x-1$ models, as $\exists p_j$ in $L_j$ where $x > j > 1$, $p_j$ can be generated by removing ($j-1$) models from $p_1$, if these ($j-1$) models all belong to $M_{x-1}$, then $p_x$ can be obtained by removing ($x-j$) models belonging to$M_{x-1}$ from $p_j$.

    For the above plan $p_j$, $\tbinom{|M_{x-1}|}{j-1}$ distinct plan $p_j$ can be obtained through $p_1$. We only need to prove $\tbinom{|M_{x-1}|}{j-1} > x-j$.

    $\tbinom{|M_{x-1}|}{j-1}$ can be expressed as 
    
    \begin{equation}
        \begin{aligned}
        \tbinom{|M_{x-1}|}{j-1} &= \tbinom{x-1}{j-1} \\
        &= \frac{(x-1)!}{(j-1)!\times(x-j)!}\\
        &= \frac{(x-1)!}{(j-1)\times(j-2)\times...\times2\times(x-j)!},
        \end{aligned}
    \end{equation}

    since 

    \begin{equation}
        \begin{aligned}
        x-1 &= \frac{(x-1)!}{(x-2)!} \\
        &= \frac{(x-1)!}{(x-2)\times...\times(x-j+1)\times(j-1)!},
        \end{aligned}
    \end{equation}
    
    and $(x - 2 \geq j-1$, then  
    \begin{equation}
        \begin{aligned}
        (x-2)\times...\times(x-j+1) \geq (j-1)\times(j-2)\times...\times2,
        \end{aligned}
    \end{equation}
    so $\tbinom{|M_{x-1}|}{j-1} \geq x-1 > x-j$. 
    
\end{proof}

\begin{theorem} \hypertarget{theoremA6}{}

For $q$, let $p_i$ be the current largest plan of $ \mathcal{B}(p_i) - \Delta t_{\{1,i\}}$ in $L_i$. $p_i$ is the plan corresponding to $q$ in $P^*$, if for $\forall p_j$ in $L_j$, it satisfies 
\begin{equation}
    \begin{aligned}
        \mathcal{B}(p_j) - \Delta t_{\{1,j\}} < \frac{1}{|M|+1-j}(\mathcal{B}(p_i) - \Delta t_{\{1,i\}}),
    \end{aligned}
\end{equation}
where $j > i$, $|M|$ is the largest number of models in the plan.


\end{theorem}

\begin{proof}

For $\forall p_{x}$ in $L_{x}$, $\exists p_{j1}, p_{j2}$ in $L_j$, $p_{j+1}$ is obtained by removing one model from $p_{j1}$ or $p_{j2}$ according to Lemma 1. Let $P^{-q}$ represent the plan combinations for $Q$ that do not include the plan corresponding to $q$, if 
    \begin{equation}
        \begin{aligned}
            \mathcal{B}(p_{j1}, P^{-q}) - \Delta t_{\{1,j1\}} < \frac{1}{|M|+1-j}(\mathcal{B}(p_i) - \Delta t_{\{1,i\}})
        \end{aligned}
    \end{equation}
and
    \begin{equation}
        \begin{aligned}
            \mathcal{B}(p_{j2}, P^{-q}) - \Delta t_{\{1,j2\}} < \frac{1}{|M|+1-j}(\mathcal{B}(p_i) - \Delta t_{\{1,i\}}), 
        \end{aligned}
    \end{equation}
then,
    \begin{equation}
        \begin{aligned}
            &\quad\mathcal{B}(p_{j+1}, P^{-q}) - \Delta t_{\{1,j+1\}} \\
            &\leq \mathcal{B}(p_{j1}, P^{-q}) + \mathcal{B}(p_{j2}, P^{-q}) - \Delta t_{\{1,j+1\}}\\
            &\leq \mathcal{B}(p_{j1}, P^{-q}) + \mathcal{B}(p_{j2}, P^{-q}) - \Delta t_{\{1,j1\}} - \Delta t_{\{1,j2\}} \\
            &< \frac{2}{|M|+1-j}(\mathcal{B}(p_i) - \Delta t_{\{1,i\}}). 
        \end{aligned}
    \end{equation}


Since $|M| \geq j+1$, then 
    \begin{equation}
        \begin{aligned}
            \mathcal{B}(p_{j+1}, P^{-q}) - \Delta t_{\{1,j+1\}} 
            &< \frac{2}{|M|+1-j}(\mathcal{B}(p_i) - \Delta t_{\{1,i\}}) \\
            &< (\mathcal{B}(p_i) - \Delta t_{\{1,i\}}). 
        \end{aligned}
    \end{equation}

Since the maximum number of layers is determined by the RL plan with the most models, denoted as $|M|$. We set $j < x \leq |M|$. In $L_x$, each plan consists of at most $|M| - x + 1$ models. For any given plan $p_x$ in $L_x$, as shown in Lemma A.5, there exist $|M|+1 - j$ plans $\{p_j1, p_j2, ..., p_{j(|M|+1 -j)}\} $ in $L_j$ from which $p_x$ is generated by removing $|M| - j$ models. Since each plan $p_j$ in $L_j$ satisfies Equation (28), therefore,

    \begin{equation}
        \begin{aligned}
            &\quad\mathcal{B}(p_{x}, P^{-q}) - \Delta t_{\{1,x\}} \\
            &\leq \sum_{k=1}{|M|+1  - j}\mathcal{B}(p_{jk}, P^{-q}) - \Delta t_{\{1,x\}}\\
            &\leq \sum_{k=1}{|M|+1  - j}\mathcal{B}(p_{jk}, P^{-q}) - \sum_{k=1}{|M|+1  - j}\Delta t_{\{1,jk\}}\\
            &< \frac{|M|+1  - j}{|M|+1-j}(\mathcal{B}(p_i) - \Delta t_{\{1,i\}}) \\
            &\leq \mathcal{B}(p_i) - \Delta t_{\{1,i\}}).
        \end{aligned}
    \end{equation}


\end{proof}

According to Theorem \hyperlink{theoremA2}{1}, plans in $L_i$($i > 1$) can all be derived from plans in $L_1$. Therefore, for each plan in $L_1$ of query $q_i$, we remove all models $m$ with $\mathcal{B}(\{m, P^{-q_i}\}) -c_t(m) > 0$, resulting in plan $p^{q_i*}_j$ with the greatest benefit in the subtree of plan $p^{q_i}_j$. Then, taking benefit into consideration, we rank $p^{q_i*}_j$ and select the largest one for $q_i$ as follows:

\begin{displaymath}
\mathcal{B}(P) - \sum_{i=1}^{|m|}c_t(m_i) - \Delta t_{\{1,j\}},
\end{displaymath}

where the current plan corresponding to $q_i$ in $P$ is $p^{q_i*}_j$.
For example, as shown in Figure \hyperlink{fig5}{5(a)}, the batch of queries $Q = \{q_1, q_2, q_3\}$.  We only need to consider all plans in the first layer $L_1$. First, we calculate the benefit $\mathcal{B}(\{m, p^{q_2}_1, p^{q_1}_5\})$ for each model $m$ within $q_3$. We find that $\mathcal{B}(\{m_3, p^{q_2}_1, p^{q_1}_5\}) -c_t(m_3) = 2*c_t(\Delta r_3) - c_t(m_3) > 0$, that is, removing $m_3$ can bring benefits to the return time of $Q$. As shown in Figure \hyperlink{fig5}{5(b)}, we remove $m_3$ to obtain $p^{q_3*}_1$, and then we iteratively perform the same operations on all other plans $p^{q_3*}_j$ and then rank to obtain the final plan combination.



\section{Experiments}
    







\subsection{Experimental Setup}

All experiments are conducted on a computer running Ubuntu 18.04 64-bit with an Intel(R) Xeon(R) E5-2620 v2 @ 2.10GHz * 24 CPU, 64 GB RAM and 2 TB disk. 


\subsubsection{\textbf{Datasets.}}
We evaluate MLego on six datasets of varying sizes, ranging from 30,000 to 30 million documents. Among them, \textit{PubMED and NYTimes}\cite{NYtimes_and_PubMED} are the most commonly used datasets in topic modeling research, with NYTimes containing lengthy texts and PubMed comprising shorter texts. \textit{Amazon Fine Food Reviews}\footnote{https://www.kaggle.com/datasets/snap/amazon-fine-food-reviews}\cite{amazon}, \textit{Wiki-movie-plot}\footnote{https://www.kaggle.com/datasets/jrobischon/wikipedia-movie-plots}, \textit{News Category}\footnote{https://www.kaggle.com/datasets/rmisra/news-category-dataset}\cite{categoryDS1, categoryDS2}, and \textit{Realnews} \cite{realnews} are real-world datasets, where each document contains various attributes such as time, score, and others, making them suitable for OLAP query generation. Notably, \textit{Realnews} contains over 30 million documents, making it ideal for scalability evaluation.

\subsubsection{\textbf{Query Workloads.}}
To the best of our knowledge, there is no off-the-shelf benchmark for model merging and model plan searching. To solve the problem, we followed \cite{vldb_Hasani} to generate two types of query workloads from the above datasets. Random workload involves queries where the query predicates are chosen at random, and OLAP workload involves queries with predicates on attributes that impose OLAP hierarchies. Table \hyperlink{workload}{2} illustrates some of them.

\begin{table}[htbp]
\caption{Some of the queries used in the experiments}\hypertarget{workload}{}
\centering
\begin{tabular}{c|c|m{4.4cm}}
\toprule
\centering\textbf{Dataset}&\textbf{$Q\#$}&\textbf{Quries}\\
\midrule
    \centering PubMed & q1 & SELECT $lda$ FROM PubMed WHERE $id$ $<$ 20423 \\ 
\midrule
    \centering NYTimes & q1 & SELECT $lda$ FROM NYTimes WHERE $id$ $<$ 42394\\ 
\midrule
    \centering Realnews & q1 & SELECT $lda$ FROM Realnewes WHERE $time$ IN '2019.01'\\ 
\toprule
\end{tabular}
\label{tab1_query}
\end{table}




We generate random query workloads from PubMed, NYTimes, Realnews, and Amazon Food Reviews, with the WHERE clause specifying the \textit{id} range. To ensure robustness, data is shuffled using 10 random seeds.  
OLAP query workloads, used for model merging and plan searching, are derived from News Category, Wiki-movie-plot, Realnews, and Amazon Food Reviews. We divide OLAP cuboids into 1\%–10\% of tuples and sample ranges based on model requirements. Queries are generated by considering all possible subsets of cuboids with the same predicate. 
A summary of datasets and query workloads is provided in Table \hyperlink{dataset}{3}.

\begin{table}[H]
    \caption{Dataset Statistics}\hypertarget{dataset}{}
    \begin{center}
    \begin{tabular}{ccc}
        \hline
        \textbf{Dataset} &\textbf{\#Docs}& \textbf{Workload Type} \\
        \hline
            Amazon Food Reviews & 568,454 & OLAP \& Random \parbox[c]{5cm}{} \\
            Realnews  & 32,797,763& OLAP\&Random  \parbox[c]{5cm}{} \\
            Wiki-movie-plot & 33,869& OLAP   \parbox[c]{5cm}{} \\
            News Category  & 202,372& OLAP  \parbox[c]{5cm}{}\\
            PubMed  & 141,043& Random  \parbox[c]{5cm}{} \\
            NYTimes  & 102,660 & Random\parbox[c]{5cm}{} \\
        \hline
    \end{tabular}
    \label{tab1_dataset}
    \end{center}
\end{table}


\subsubsection{\textbf{Evaluation Metrics.}}
We evaluate our methods against the baselines across two criteria: time and model accuracy. 

For model accuracy, we employ a commonly used metric in topic modeling research - Log Predictive Probability (lpp). Under this metric, a higher score is better, as a better model will assign a higher probability to the held-out words.



In addition, We followed \cite{vldb_Hasani} to introduce two metrics: \textbf{(1) Speedup Ratio(SR)}, and \textbf{(2) Difference in Performance(DP)}. These metrics are based on time and lpp, where SR is defined as the ratio of time taken to build a model by model merging and DP measures the difference in lpp between the exact and approximate models.



\subsubsection{\textbf{Baselines.}}
We select several baselines for comparison. We use abbreviations MGS and MVB for the VB-based model merging method and GS-based model merging method, respectively. We compare our LDA model merging methods against the following methods. (i) \textbf{ORIG} executes analytic queries by building the ML model from scratch. (ii) \textbf{LDA*}\cite{LDA*}, deployed at Tencent for internal topic modeling services, integrates multiple recent samplers, including AliasLDA\cite{AliasLDA}, F+LDA\cite{F+LDA}, LightLDA\cite{LightLDA}, and WarpLDA\cite{WarpLDA}. (iii) \textbf{OGS}\cite{jmlr_online_17} was an algorithm that uses a novel online inference method for large-scale latent variable models.


Next, we evaluate our plan searching algorithms by comparing them with (a) \textbf{NAI} approach (generate-and-rank: enumerate all plans and calculate the score, then rank them.) and (b) \textbf{GRA}\cite{vldb_Hasani} method (A baseline that constructs a Directed Acyclic Graph (DAG) based on the WHERE clause of queries and utilizes the shortest path strategy to get the optimal plan.). We use the abbreviations PSOA and PSOA++ for our plan searching algorithm and optimization method in Section 5.B.

For our model merging methods and baselines about LDA, we set the number of topics to 100, the maximum number of iterations to 100, and keep other settings at their default configurations. We run MGS and MVB in single-threaded mode for all experiments. To ensure a fair comparison, we configure LDA* to run on Hadoop with one worker.

Before conducting the experiments for efficiency of query plan searching, we derive a performance loss function based on model merging experiments. It is worth noting that, during the analysis of the performance loss, we assume it monotonically decreases with the number of model merging, a hypothesis we validate through experiments. Our experiments regarding query plan searching efficiency remain unaffected by the performance loss function.

\subsection{Experimental Results}

\subsubsection{\textbf{Effect of Model Merging.}} 
We evaluate whether the performance loss caused by model merging aligns with the assumption in Section 5.2. We consider the impact of different model numbers on the model merging effects. For each query, we build a model on the scratch corresponding to the query, and then randomly divide the data into 2, 3, 4 until 30 partitions, train models separately, and then merge models. As shown in Figure \hyperlink{fig6}{6}, for both random and OLAP query workloads, performance tends to decrease monotonically with an increasing number of model merging (i.e., \#candidate models), and MGS demonstrates lower performance loss than MVB. Notably, lpp for MGS reaches its maximum value of -7.52 after the first merging, and it gradually decreases with an increasing number of model merging. This is in line with our assumption, and we can consider the plan with a merge count of 1 as the first layer of the $l_p$ list in plan searching.

\begin{figure}[htbp]
\centerline{\includegraphics[width=0.5\textwidth]{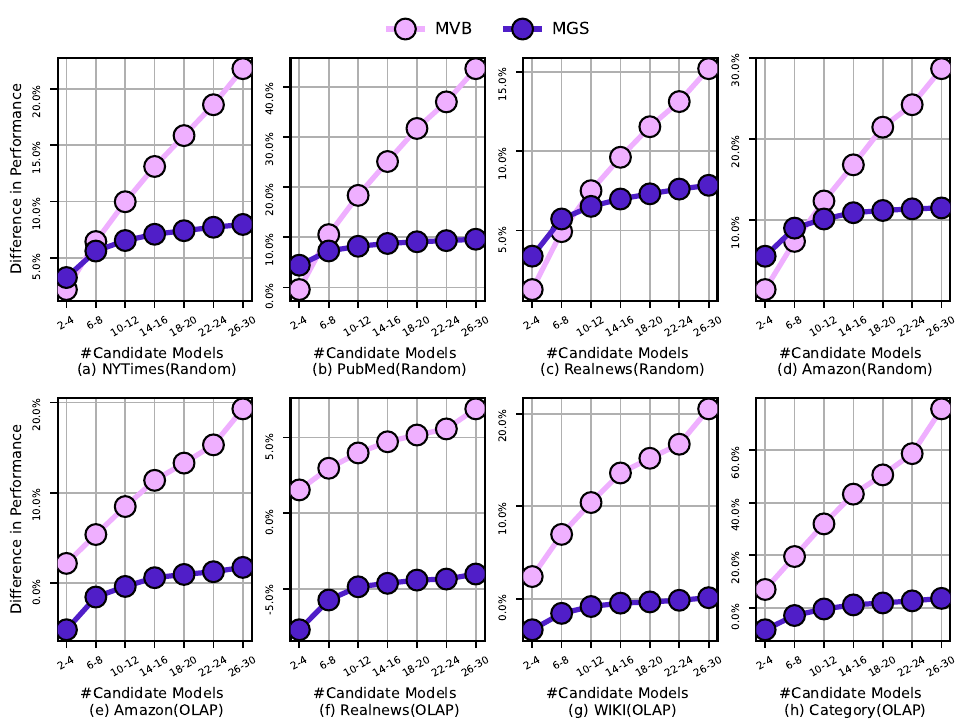}}
\caption{Evaluating approaches for building an approximate ML model for \textbf{Random} and \textbf{OLAP} query workloads.}
\label{fig_rq1_all_pl}
\end{figure}


\subsubsection{\textbf{Efficiency of Model Merging.}}To study the efficiency of our model merging methods, we compared our method with the baselines in terms of model building time.
We computed SR of our two model merging methods relative to LDA* and OGS under the condition of ample materialized models, as shown in Figure \hyperlink{fig7}{7}. The merging time for both model merging methods was very close and quite short. Even when compared to distributed model training systems like LDA*, MLego showed a significant improvement in model construction time due to the reuse of materialized models. As shown in Figure \hyperlink{fig8}{8}, MLego scales well, demonstrating greater advantages as data size increases, making it suitable for interactive analysis on large datasets.


\begin{figure}[htbp]
\centerline{\includegraphics[width=0.5\textwidth]{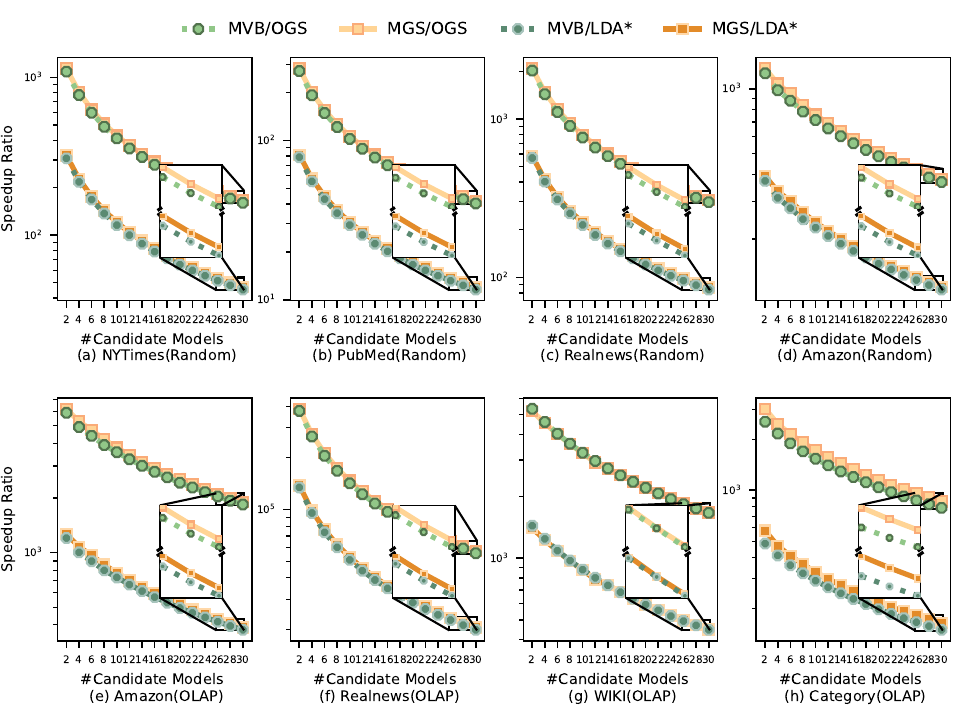}}
\caption{Efficiency of approaches for building an approximate ML model for \textbf{Random} and \textbf{OLAP} query workloads.}
\label{fig_efficiency_model_merge_all_SR}
\end{figure}

\begin{figure}[ht]
\centering
\subfloat[Amazon]{\includegraphics[width=1.6in]{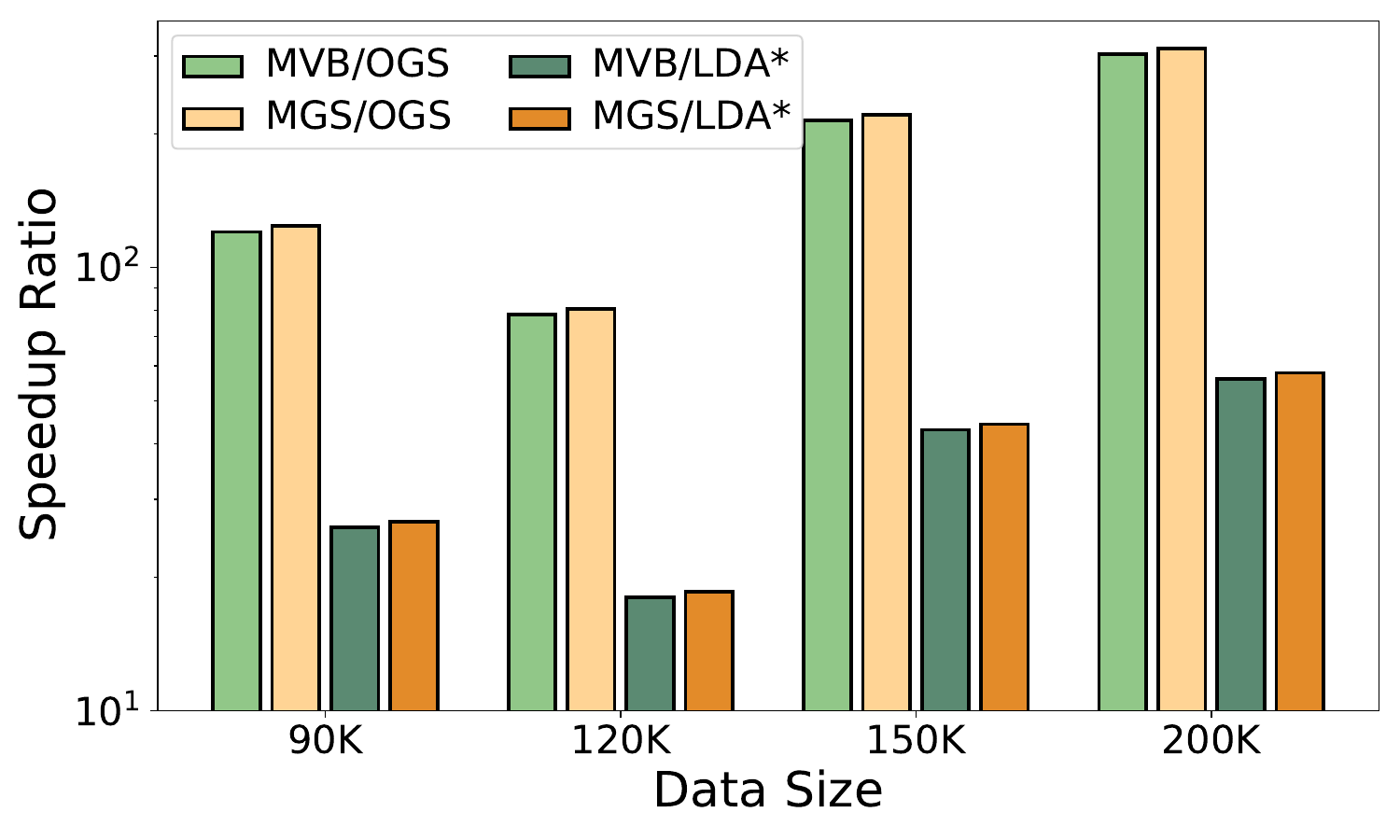}%
\label{fig_first_case1}}
\hfil
\centering
\subfloat[Realnews]{\includegraphics[width=1.6in]{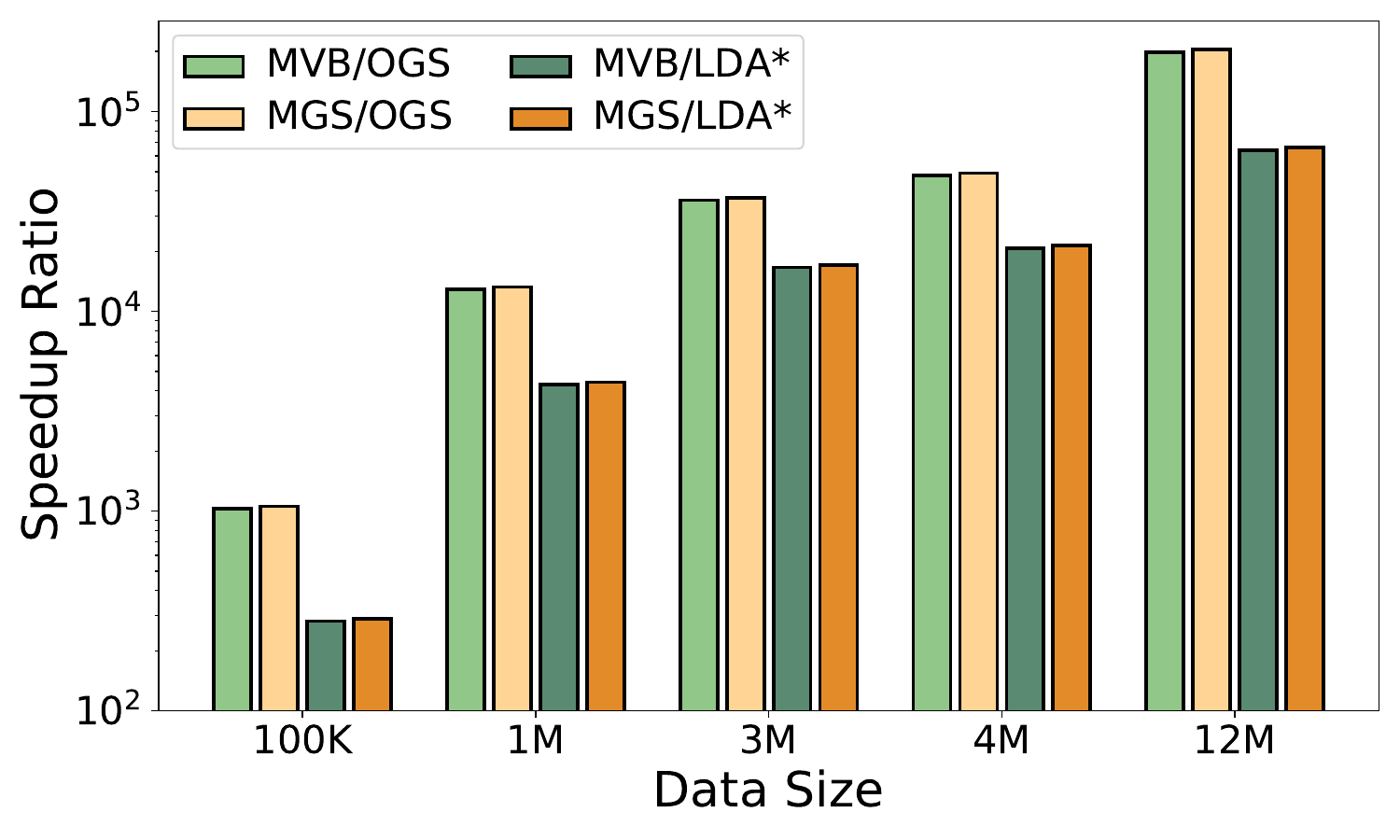}%
\label{fig_second_case1}}
\caption{Scalability of model merging methods}\hypertarget{fig8}{}
\label{fig_7}
\end{figure}

\begin{figure}
\centering
\subfloat[Amazon]{\includegraphics[width=1.6in]{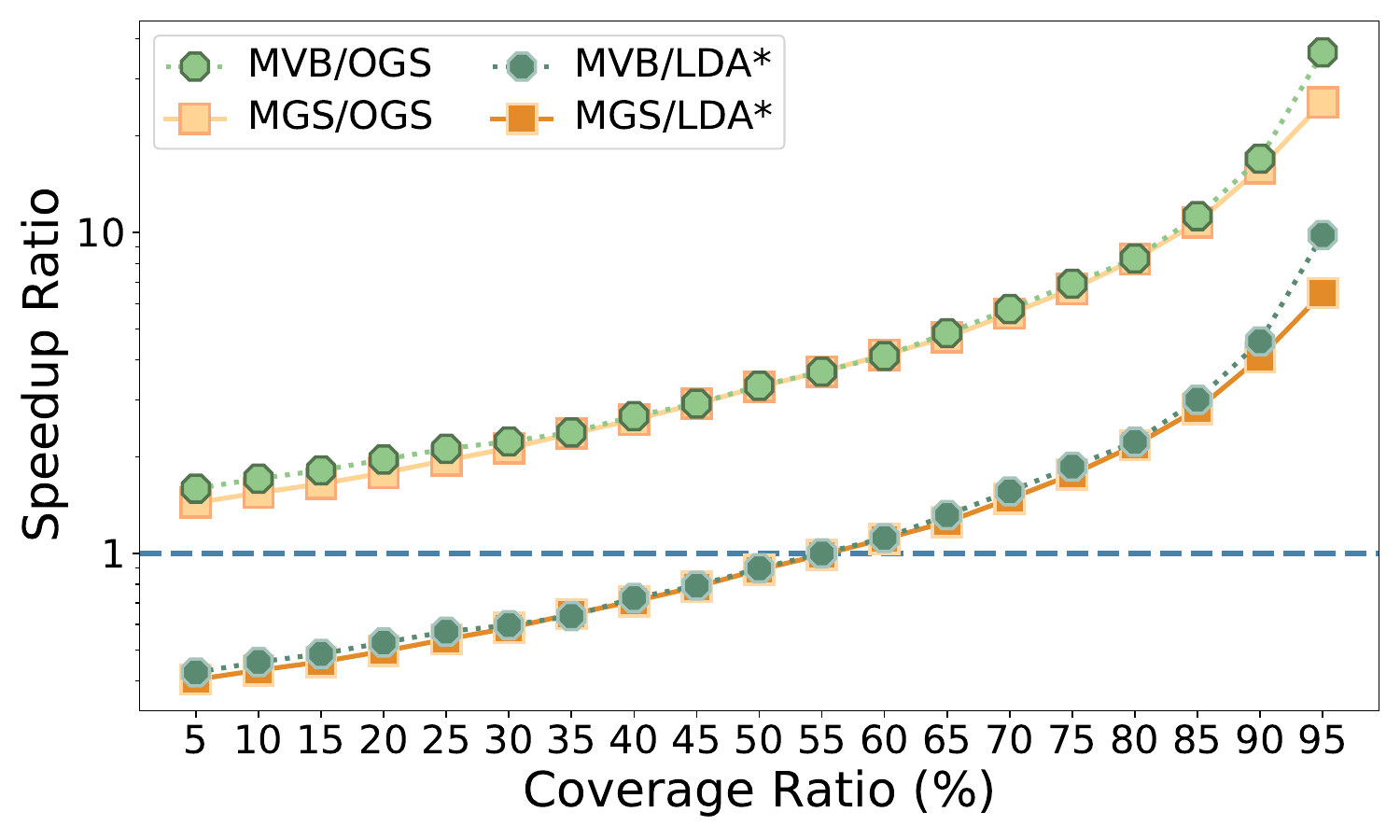}%
\label{fig_first_case2}}
\hfil
\centering
\subfloat[Realnews]{\includegraphics[width=1.6in]{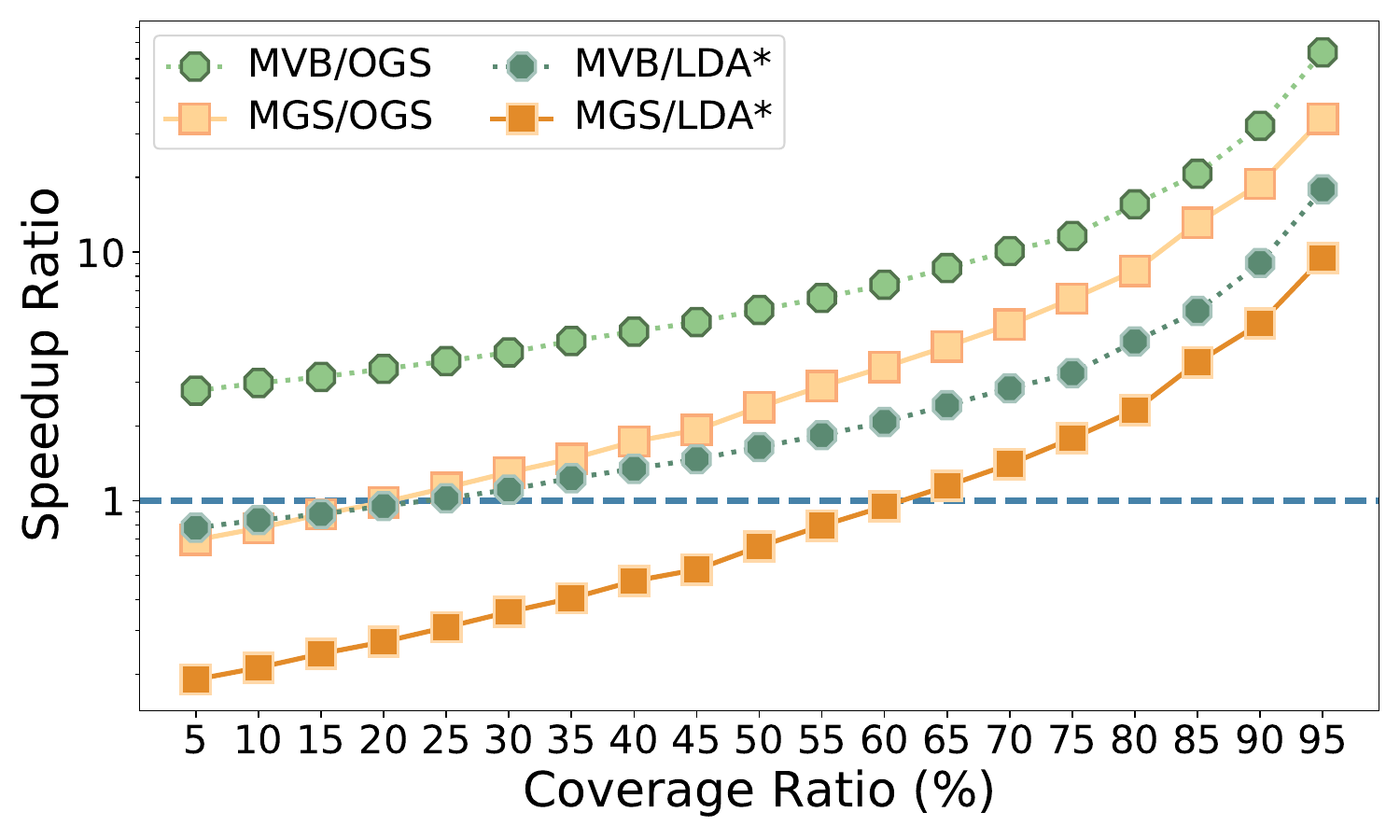}%
\label{fig_second_case2}}
\centering
\caption{Evaluating impact of Coverage Ratio on building an approximate model.}\hypertarget{fig9}{}
\label{fig_2}
\end{figure}

Furthermore, we studied the impact of coverage ratio on the efficiency of MLego. The coverage ratio reflects the extent to which we can answer queries using materialized models. A 100\% coverage ratio indicates that we can fully respond with pre-built models, whereas a 0\% coverage ratio means we need to construct the model from scratch. Figure \hyperlink{fig9}{9} shows SR under different coverage ratios. We can observe that the slowest MLego is equal to the baseline efficiency at a coverage ratio of 55\%, and MLego based on MGS outputs $m^*$ more quickly than LDA*. As the coverage ratio continues to increase, the model construction time decreases. When the coverage ratio reaches 100\%, the model can be constructed in milliseconds. However, at this point, the time spent on plan searching becomes a significant portion of the entire workflow. That is why we need to address the efficiency of plan searching.



\begin{figure*}[!t]
\centering
\subfloat[Amazon]{\includegraphics[width=1.6in]{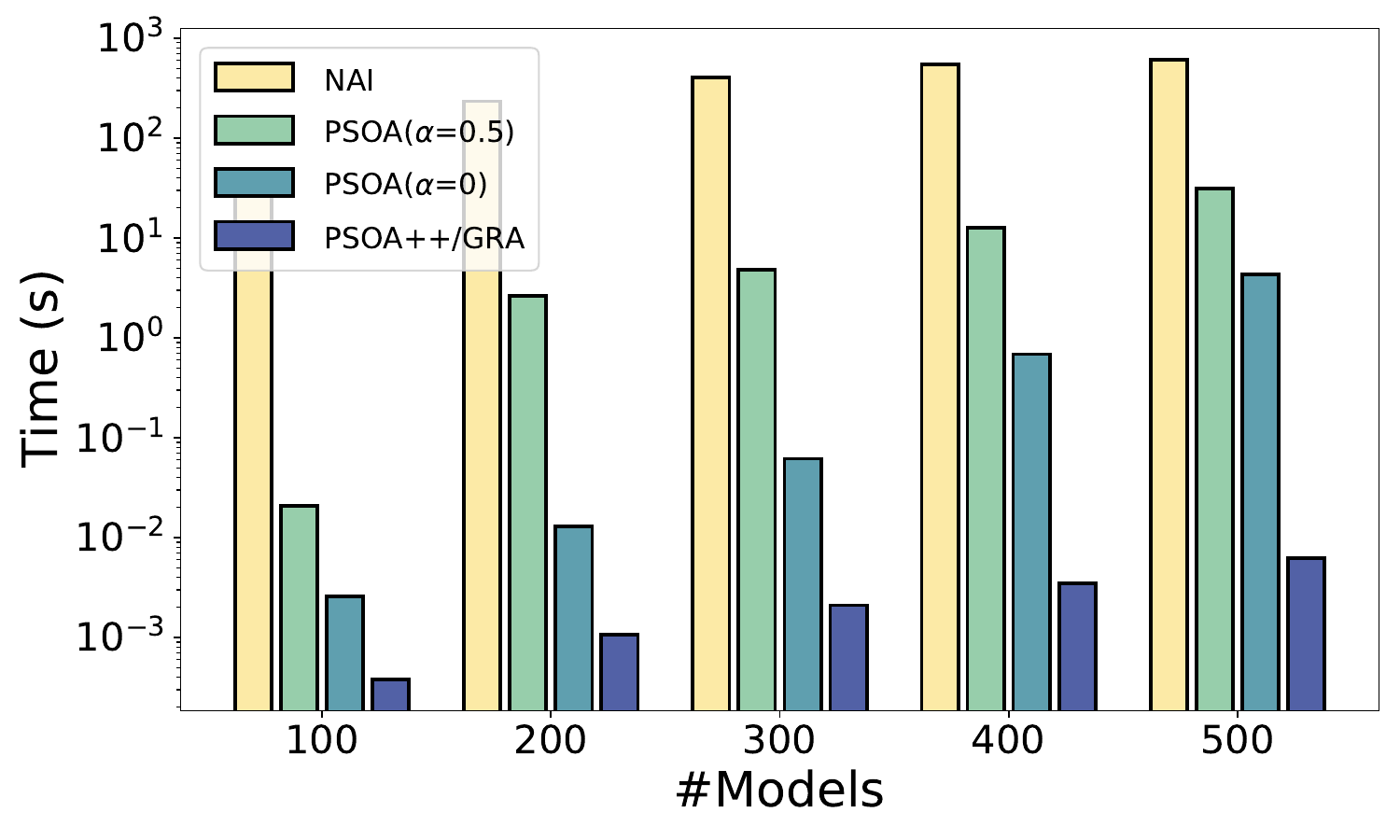}%
\label{fig_first_case3}}
\hfil
\centering
\subfloat[Category]{\includegraphics[width=1.6in]{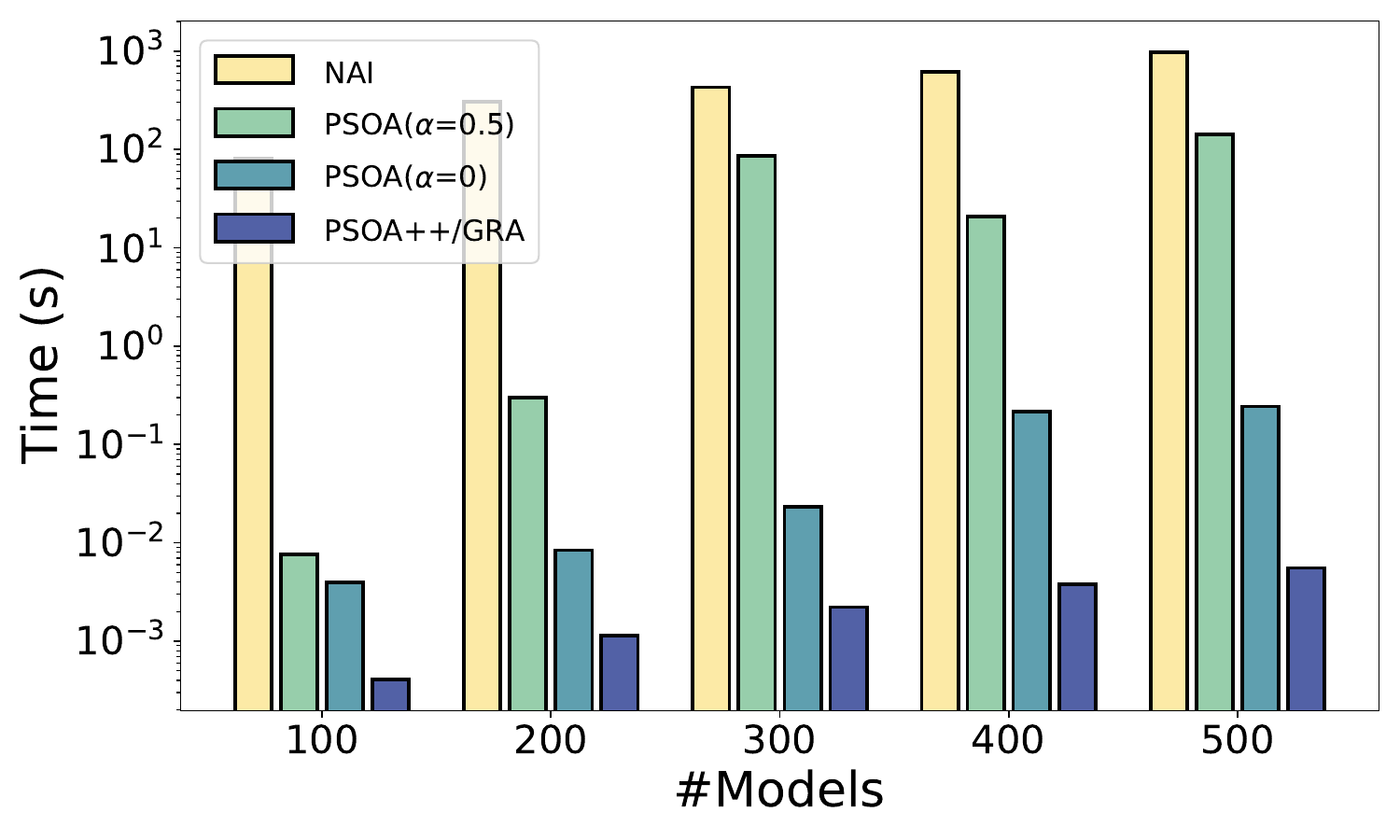}%
\label{fig_second_case3}}
\hfil
\centering
\subfloat[Realnews]{\includegraphics[width=1.6in]{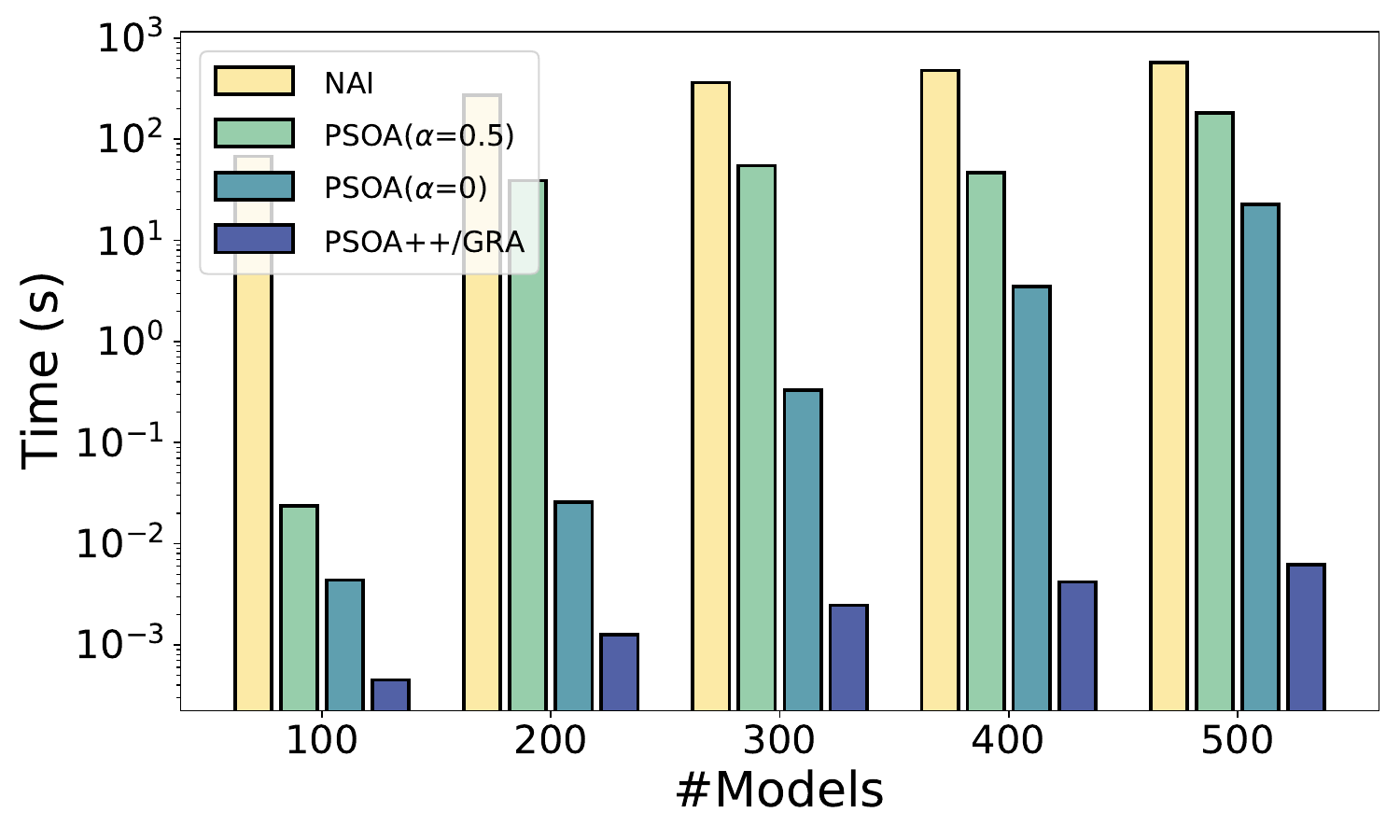}%
\label{fig_second_case4}}
\hfil
\centering
\subfloat[WIKI]{\includegraphics[width=1.6in]{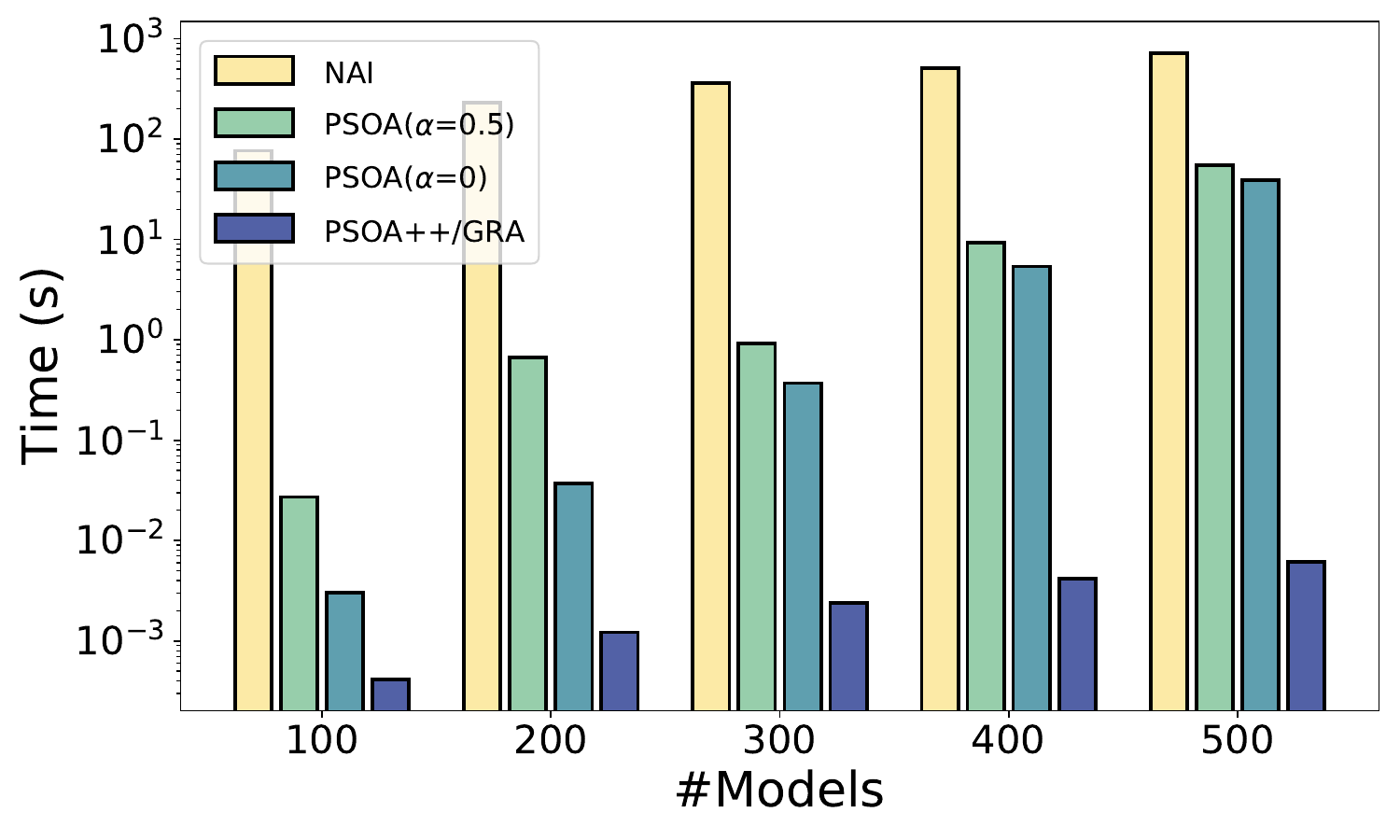}%
\label{fig_second_case5}}
\caption{Efficiency of plan searching algorithm}\hypertarget{fig10}{}
\label{fig_sim}
\end{figure*}

\subsubsection{\textbf{Time Reduction of Plan Searching Algorithm.}}

MLego searches for the optimal plan in the model set based on the weight parameter $\alpha$. We constructed five model sets for each OLAP query workload, and we evaluated the efficiency of plan searching on them. As shown in Figure \hyperlink{fig10}{10}, our plan searching method PSOA exhibits exponential speedup compared to NAI. In addition, when $\alpha$ is 0, our method shows a significant improvement in speed since it only needs to compute for two lists. Furthermore, when the maximum number of models in the RL plan satisfies Theorem \hyperlink{theorem5.2}{3} or Theorem \hypertarget{theorem4}{4}, the problem becomes finding $p^*$ that maximizes model coverage over data, and our method PSOA++ aligns with GRA where GRA is only applicable to this situation.

We further explore the factors affecting the efficiency of PSOA. In Figure \hyperlink{fig11}{11}, as the number of models within a query range increases, our method still provides a substantial speedup, enabling MLego to deliver rapid plan searching capabilities in support of model merging. In addition, we evaluate the impact of $\alpha$ on PSOA, as shown in Figure \hyperlink{fig12}{12}. When alpha is set to 0.4, the search time is the slowest because it considers both time cost and performance loss, which have opposite directions in hierarchical plan generation. However, due to the "push down" operation, PSOA can still quickly compute a potentially optimal plan compared to other methods. 




\subsubsection{\textbf{Evaluation of Batch Query Optimization.}}
We evaluate the cost and effectiveness of batch query optimization on the largest model set corresponding to each dataset. For each query workload, we compute the batch query optimization cost and benefit for all possible query combinations. As shown in Figure \hyperlink{fig14}{14(a)}, we use the word count to represent the benefit, i.e., the reduction of training costs. Batch size represents the number of queries in a batch. As the batch size increases, the benefit also increases because the common data among all queries in the batch can be reused more frequently. Conversely, as shown in Figure \hyperlink{fig13}{13(a)}, when the batch size increases, the cost of batch query optimization also increases because the searching time linearly grows with the number of queries in the batch. Furthermore, we evaluate the impact of the number of candidate models covered by queries in the batch on the cost, as shown in Figure \hyperlink{fig13}{13(b)}. When \#candidate models per query is 30, the cost remains small regardless of the batch size. As the model coverage increases, the cost of batch query optimization also increases because the number of plans for each query in $L_1$ gradually grows. However, compared to the benefit of batch query optimization, the cost is highly cost-effective, and as the cost increases, the corresponding benefit also increases. Compared to the benefit brought by batch query optimization, the optimization cost is relatively small, as shown in Figure \hyperlink{fig14}{14(b)}. Additionally, as the number of models in each query increases, the benefit increases, and the searching cost becomes a smaller proportion.

\begin{figure}[ht]
\centering
\subfloat[Amazon]{\includegraphics[width=1.6in]{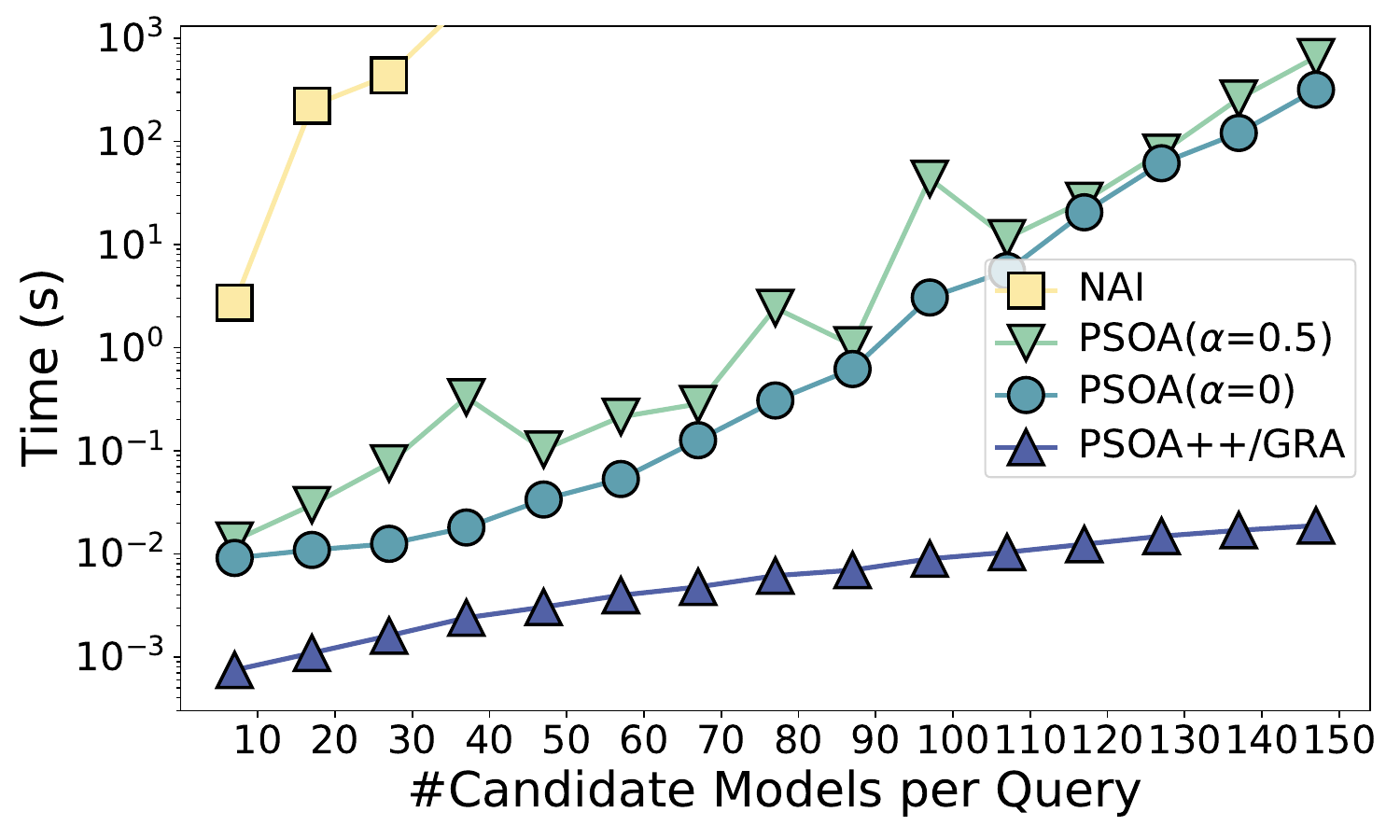}%
\label{fig_first_case6}}
\hfil
\centering
\subfloat[Realnews]{\includegraphics[width=1.6in]{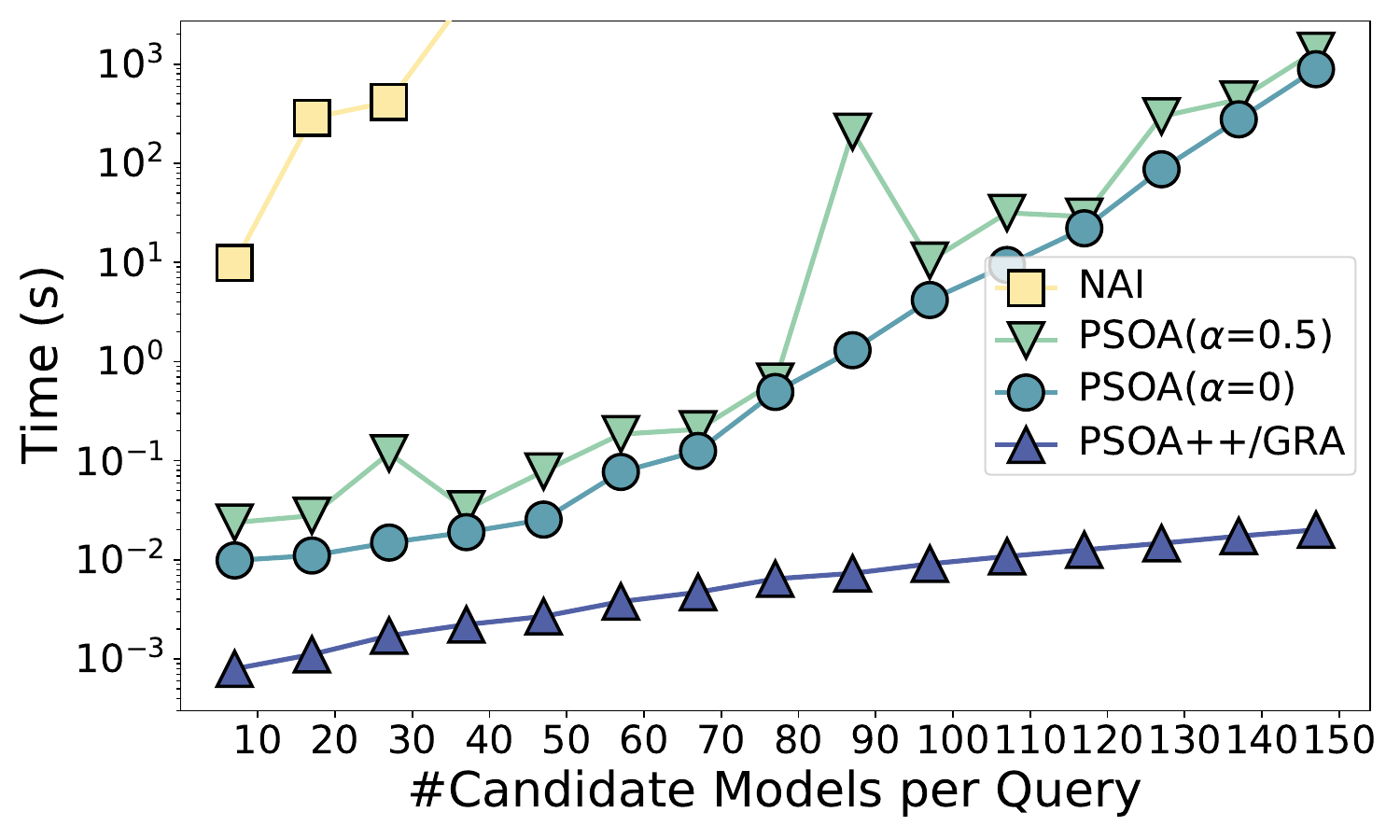}%
\label{fig_second_case6}}
\caption{Evaluating the impact of \#Candidate Models per query on plan searching}\hypertarget{fig11}{}
\label{fig_3}
\end{figure}

\begin{figure}[ht]
\centering
\subfloat[Amazon]{\includegraphics[width=1.6in]{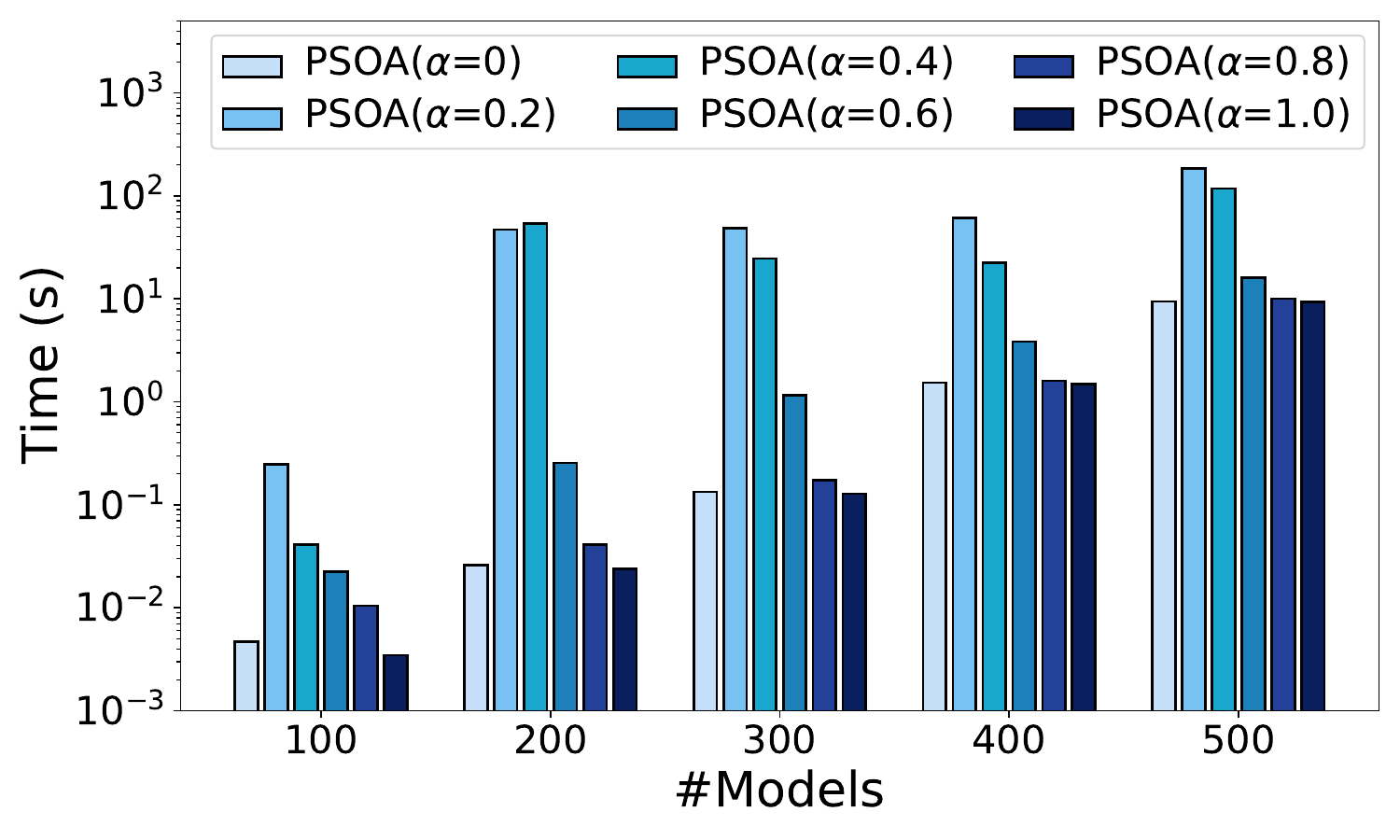}%
\label{fig_first_case7}}
\hfil
\centering
\subfloat[Realnews]{\includegraphics[width=1.6in]{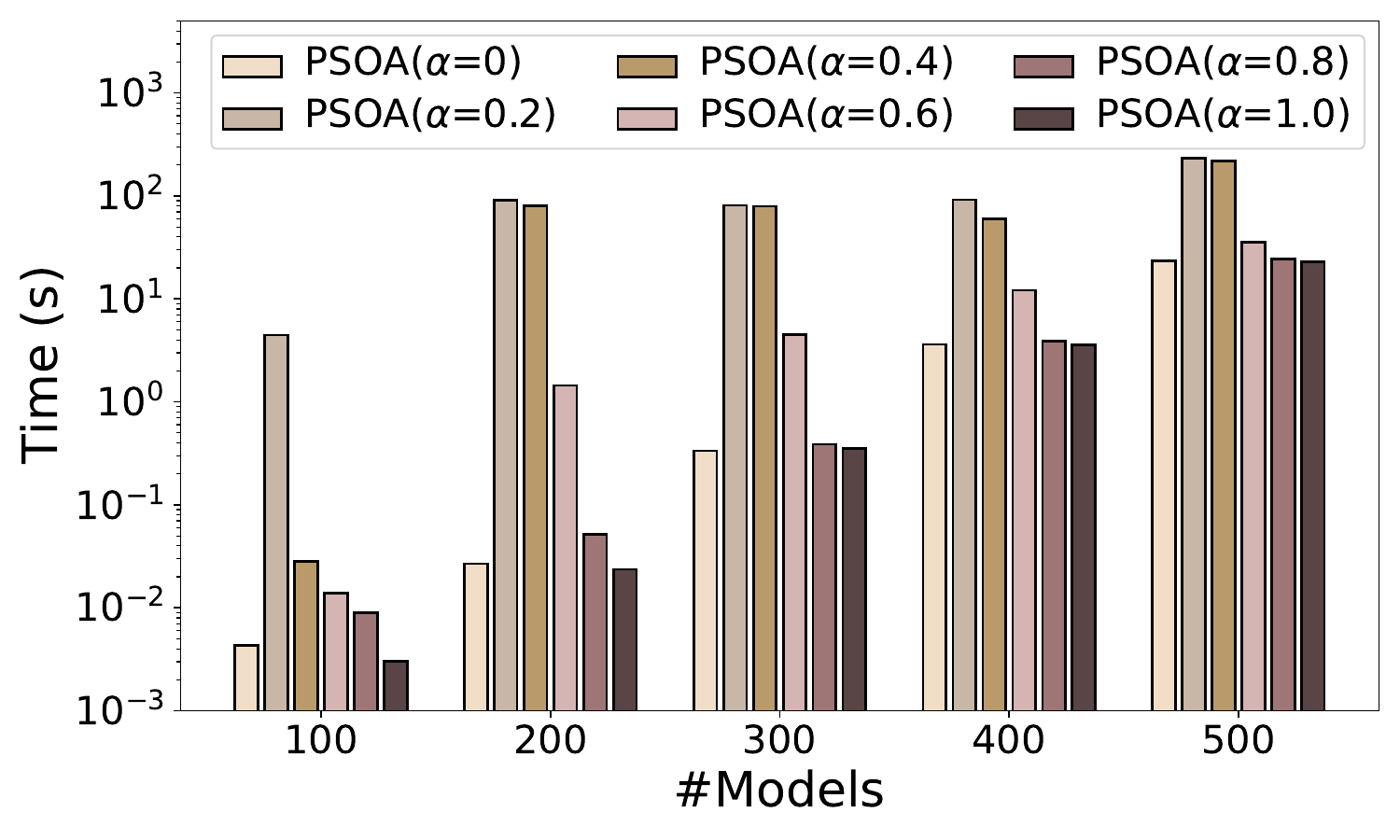}%
\label{fig_second_case7}}
\caption{Evaluating the impact of the weight parameter on plan searching}\hypertarget{fig12}{}
\label{fig_4}
\end{figure}

\begin{figure}[ht]
\centering
\subfloat[]{\includegraphics[width=1.6in]{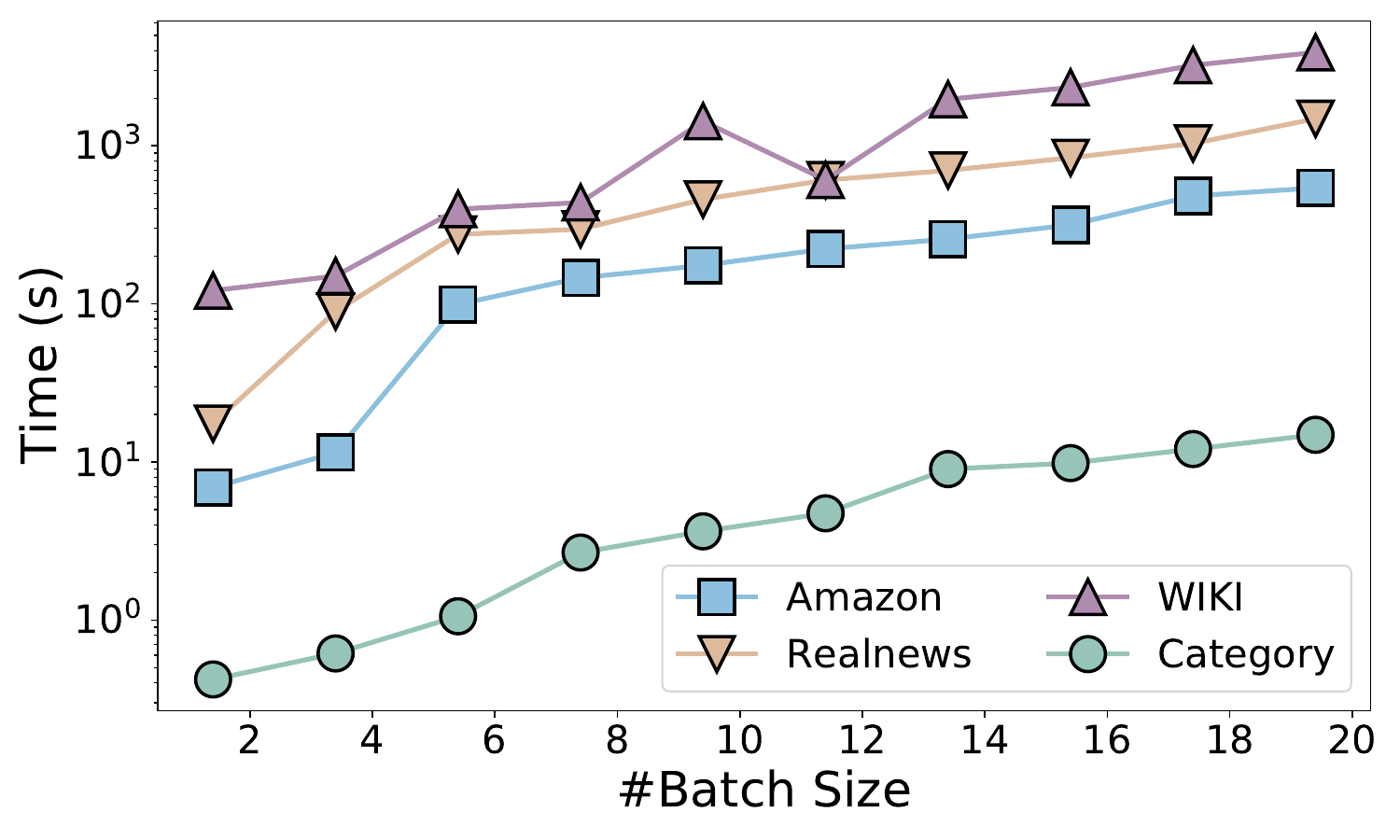}%
\label{fig_first_case8}}
\hfil
\centering
\subfloat[]{\includegraphics[width=1.6in]{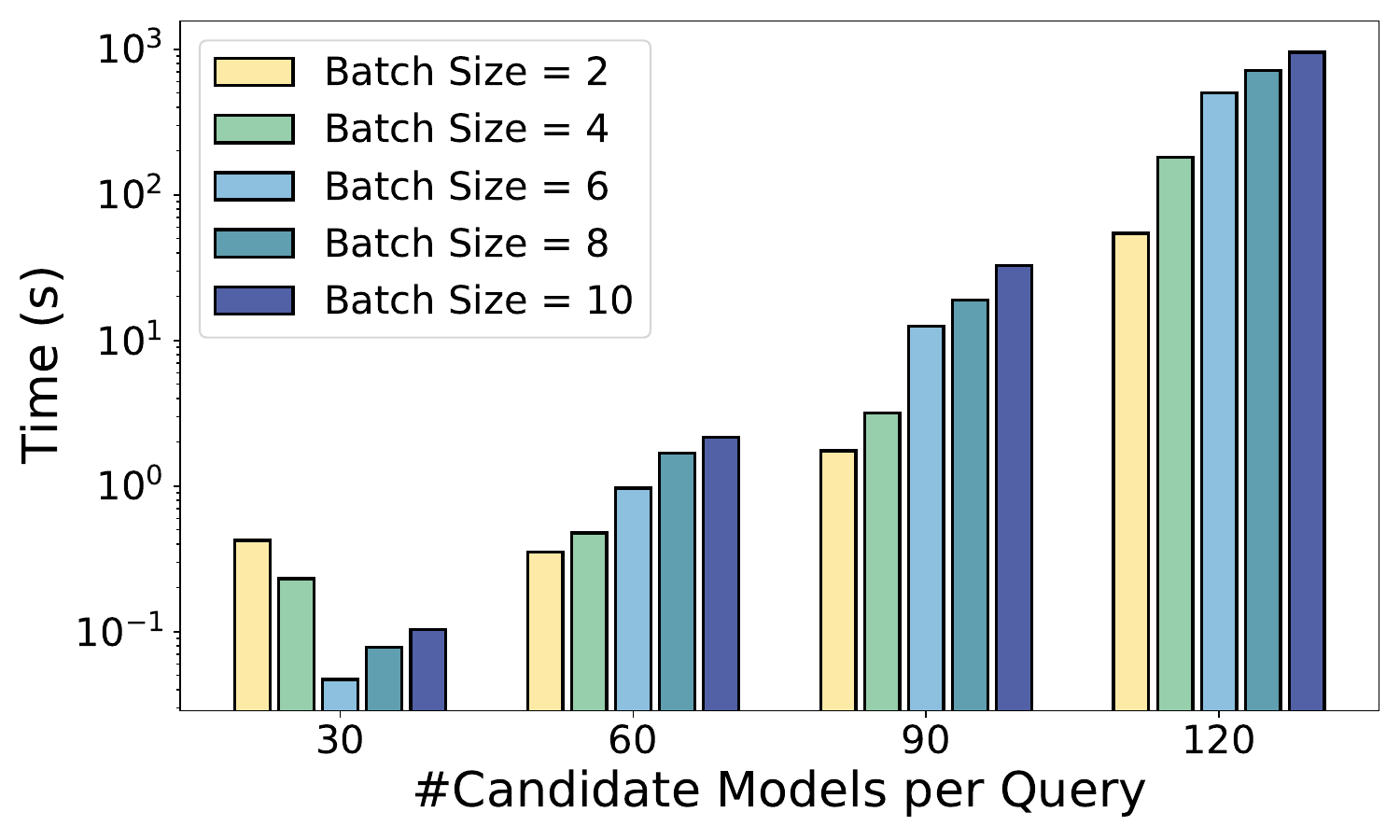}%
\label{fig_second_case8}}
\caption{Efficiency of batch query plan searching}\hypertarget{fig13}{}
\label{fig_5}
\end{figure}

\begin{figure}[ht]
\centering
\subfloat[Benefit]{\includegraphics[width=1.6in]{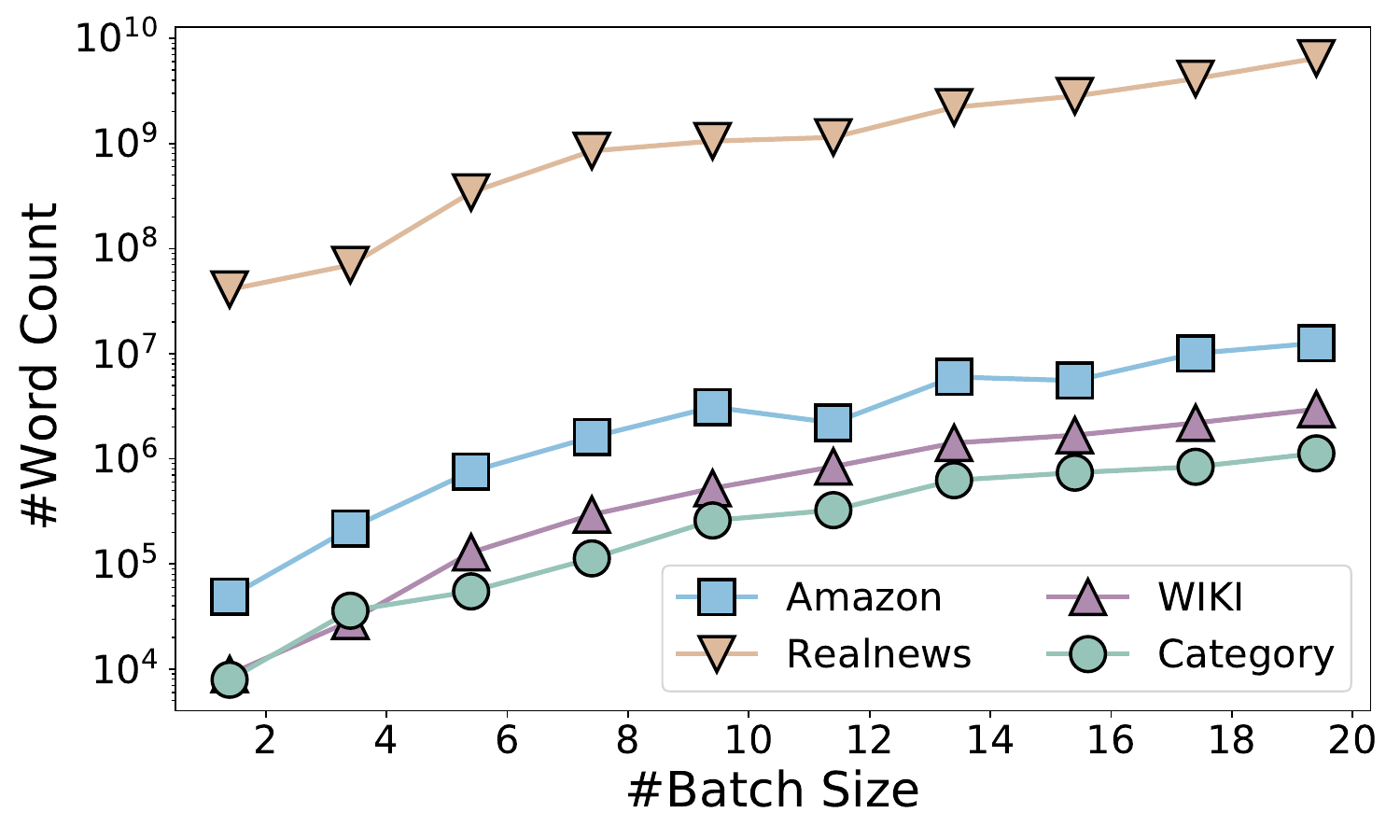}%
\label{fig_first_case9}}
\hfil
\centering
\subfloat[Benefit vs Cost]{\includegraphics[width=1.6in]{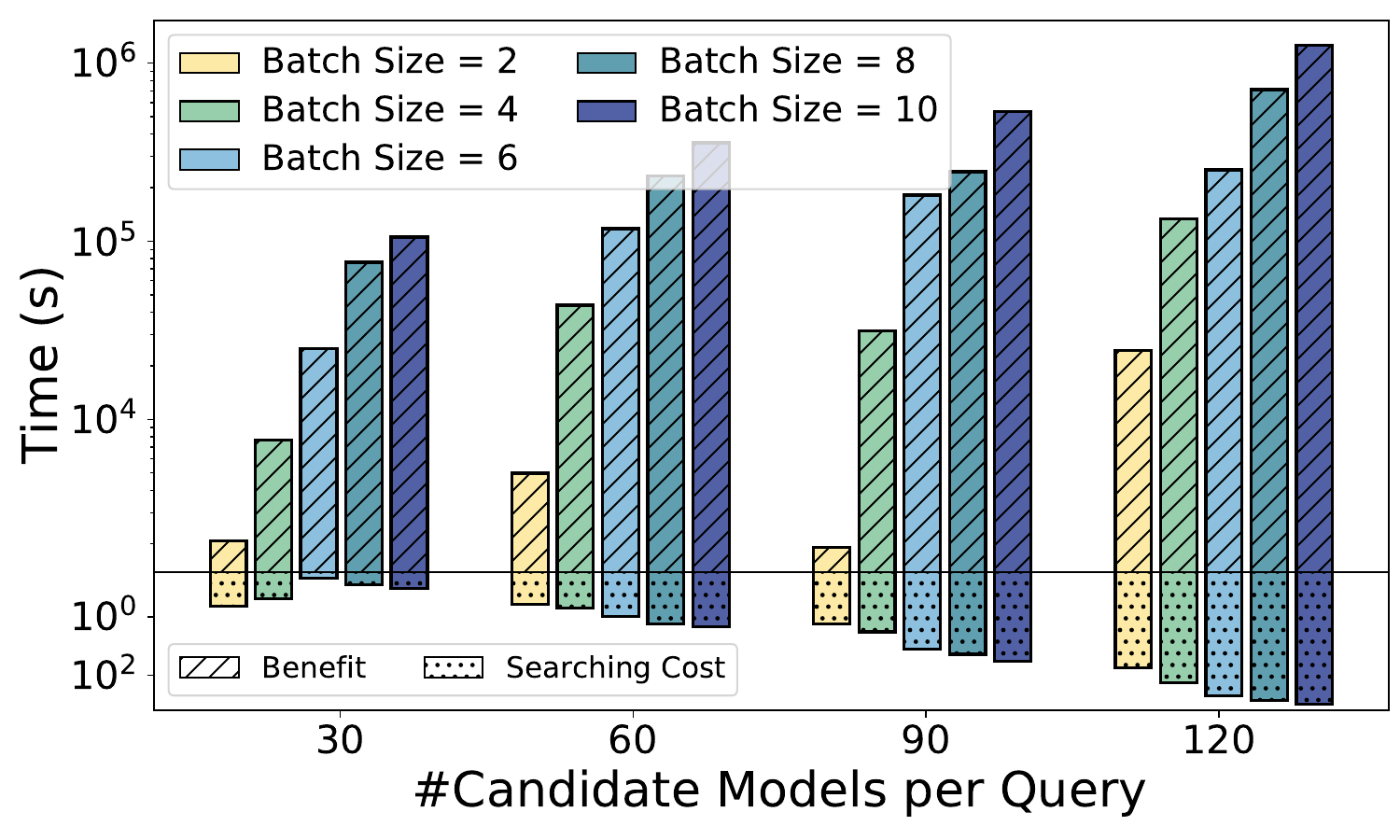}%
\label{fig_second_case9}}
\caption{Benefit and Cost of batch query optimization}\hypertarget{fig14}{}
\label{fig_6}
\end{figure}

\subsection{Usage Scenario}

Here, we present a usage scenario to demonstrate MLego's topic modeling building capabilities. 

To effectively showcase the model-building capabilities of MLego, we augment the Realnews dataset by adding two additional columns—longitude and latitude—and generate random values for these attributes. The resulting dataset, referred to as the Augmented Realnews, will be utilized for illustrating various usage scenarios. Our primary objective is to analyze the geographical distribution of news topics. For example, we aim to explore the topics of news articles associated with areas near the Louvre Museum. As illustrated in the figure, we select the Augmented dataset and focus on the geographical attributes, namely longitude and latitude, to analyze the topic distribution within specific regions. Given that our goal is to explore the topics related to the area around the Louvre, without a specific target or hypothesis, this scenario may involve multiple rounds of topic modeling and model construction for various subsets of the data.

In setting the alpha parameter, we prioritize optimizing training time. Within the MLego framework, users can not only visualize the distribution of document collections across different geographic areas but also examine the topic distribution of pre-trained models within those regions. Upon initiating the analysis, MLego generates the topic distributions for the selected geographical area, taking into account the user-specified preferences, including the alpha setting. The resulting distribution can be visualized, revealing, for instance, a topic characterized by terms such as “teriyaki, BBQ, beef, meat.” As discussed in Section 2, MLego facilitates interactive analysis by leveraging relationships among existing models, enabling rapid model construction and adaptation based on user input. This approach is orthogonal to previous works such as \cite{topicLens, B9, B11, B13, B14, B22, B44}, where users can perform granular operations—ranging from document-level to topic-level—on the results returned by MLego. In combination, these features provide a comprehensive framework for interactive analysis of text topics within geographic contexts.

\begin{figure}[H]
\centerline{\includegraphics[width=0.47\textwidth]{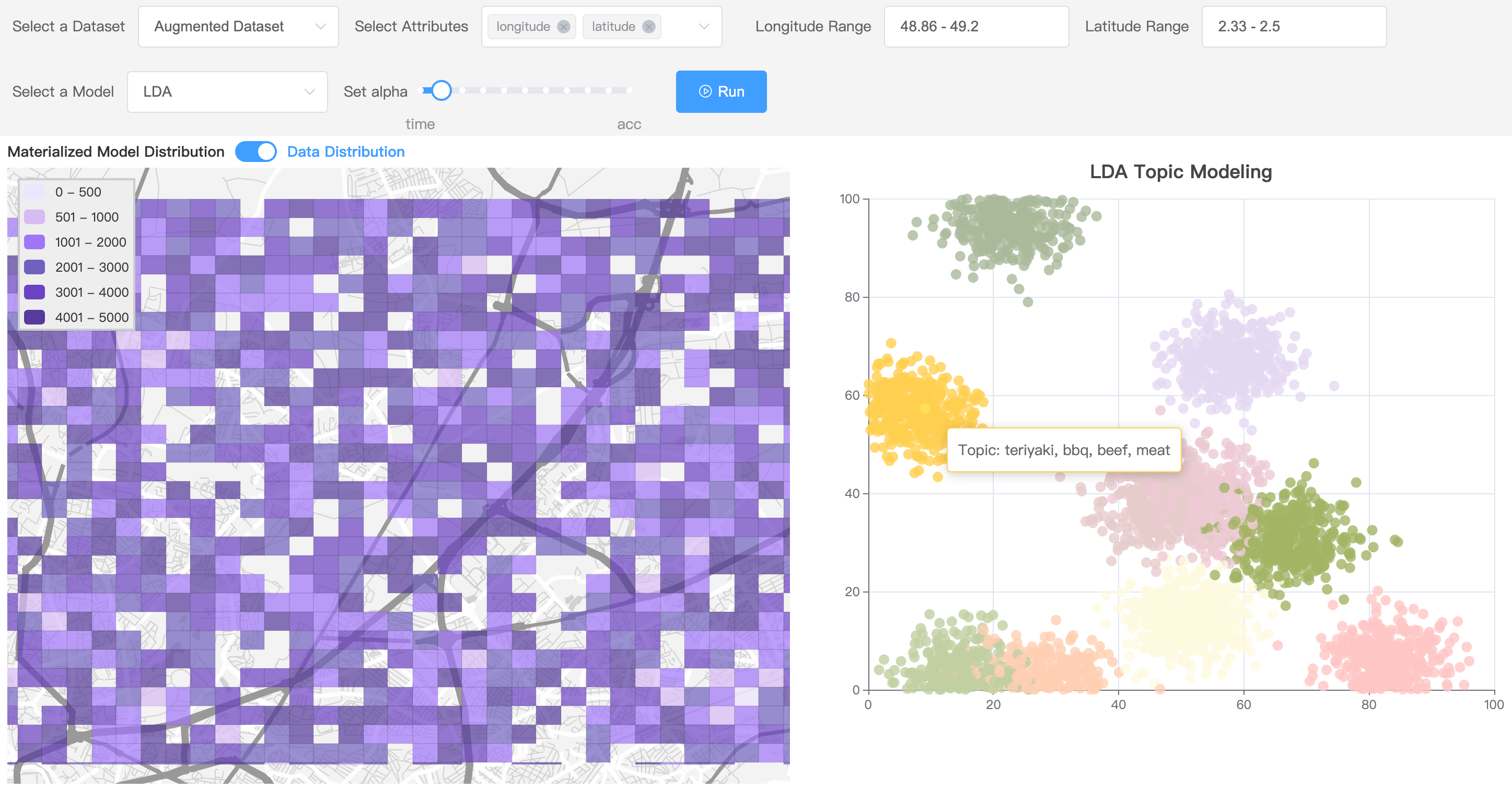}}
\caption{Usage Scenario}\hypertarget{fig4}{}
\label{fig_usage_scenario}
\end{figure}

\section{Conclusion}
In this paper, we addressed the challenge of efficiently supporting interactive, query-driven topic modeling. We proposed MLego, a novel analytic query framework that accelerates real-time topic exploration by reusing pre-built models and employing a hierarchical plan search strategy for optimal query execution. Additionally, we introduced query reordering techniques to further optimize batch queries, reducing computational overhead. To evaluate MLego, we conducted extensive experiments on real-world datasets, demonstrating its superior query efficiency and scalability. Integrated into a visual analytics prototype, MLego enables real-time, user-driven exploration of large-scale textual datasets, complementing existing visual analytics approaches that primarily focus on model steering. Our work bridges the gap between interactive data analysis and scalable topic modeling, providing an effective solution for real-time text exploration.

\bibliographystyle{IEEEtran}
\bibliography{IEEEabrv, ref}

\begin{thebibliography}{10}
\providecommand{\url}[1]{#1}
\csname url@samestyle\endcsname
\providecommand{\newblock}{\relax}
\providecommand{\bibinfo}[2]{#2}
\providecommand{\BIBentrySTDinterwordspacing}{\spaceskip=0pt\relax}
\providecommand{\BIBentryALTinterwordstretchfactor}{4}
\providecommand{\BIBentryALTinterwordspacing}{\spaceskip=\fontdimen2\font plus
\BIBentryALTinterwordstretchfactor\fontdimen3\font minus \fontdimen4\font\relax}
\providecommand{\BIBforeignlanguage}[2]{{%
\expandafter\ifx\csname l@#1\endcsname\relax
\typeout{** WARNING: IEEEtran.bst: No hyphenation pattern has been}%
\typeout{** loaded for the language `#1'. Using the pattern for}%
\typeout{** the default language instead.}%
\else
\language=\csname l@#1\endcsname
\fi
#2}}
\providecommand{\BIBdecl}{\relax}
\BIBdecl

\bibitem{topicLens}
M.~Kim, K.~Kang, D.~Park, J.~Choo, and N.~Elmqvist, ``Topiclens: Efficient multi-level visual topic exploration of large-scale document collections,'' \emph{IEEE Transactions on Visualization and Computer Graphics}, vol.~23, no.~1, pp. 151--160, 2017.

\bibitem{A9B8}
J.~Choo, C.~Lee, C.~K. Reddy, and H.~Park, ``Utopian: User-driven topic modeling based on interactive nonnegative matrix factorization,'' \emph{IEEE Transactions on Visualization and Computer Graphics}, vol.~19, no.~12, pp. 1992--2001, 2013.

\bibitem{A23}
B.~Gretarsson, J.~O’Donovan, S.~Bostandjiev, T.~H\"{o}llerer, A.~Asuncion, D.~Newman, and P.~Smyth, ``Topicnets: Visual analysis of large text corpora with topic modeling,'' \emph{ACM Trans. Intell. Syst. Technol.}, vol.~3, no.~2, Feb. 2012.

\bibitem{A32B32}
H.~Lee, J.~Kihm, J.~Choo, J.~T. Stasko, and H.~Park, ``ivisclustering: An interactive visual document clustering via topic modeling,'' \emph{Computer Graphics Forum}, vol.~31, 2012.

\bibitem{Architext}
H.~Kim, B.~Drake, A.~Endert, and H.~Park, ``Architext: Interactive hierarchical topic modeling,'' \emph{IEEE Transactions on Visualization and Computer Graphics}, vol.~27, no.~9, pp. 3644--3655, 2021.

\bibitem{B9}
J.~Chuang, Y.~Hu, A.~Jin, J.~D. Wilkerson, D.~A. McFarland, C.~D. Manning, and J.~Heer, ``Document exploration with topic modeling: Designing interactive visualizations to support effective analysis workflows,'' in \emph{NIPS Workshop on Topic Models: Computation, Application, and Evaluation}, 2013.

\bibitem{B11}
W.~Dou, L.~Yu, X.~Wang, Z.~Ma, and W.~Ribarsky, ``Hierarchicaltopics: Visually exploring large text collections using topic hierarchies,'' \emph{IEEE Transactions on Visualization and Computer Graphics}, vol.~19, no.~12, pp. 2002--2011, 2013.

\bibitem{B13}
M.~El-Assady, R.~Sevastjanova, F.~Sperrle, D.~Keim, and C.~Collins, ``Progressive learning of topic modeling parameters: A visual analytics framework,'' \emph{IEEE Transactions on Visualization and Computer Graphics}, vol.~24, no.~1, pp. 382--391, 2018.

\bibitem{B14}
M.~El-Assady, F.~Sperrle, O.~Deussen, D.~Keim, and C.~Collins, ``Visual analytics for topic model optimization based on user-steerable speculative execution,'' \emph{IEEE Transactions on Visualization and Computer Graphics}, vol.~25, no.~1, pp. 374--384, 2019.

\bibitem{B22}
E.~Hoque and G.~Carenini, ``Interactive topic modeling for exploring asynchronous online conversations: Design and evaluation of convisit,'' \emph{ACM Trans. Interact. Intell. Syst.}, vol.~6, no.~1, Feb. 2016.

\bibitem{B44}
A.~Smith, V.~Kumar, J.~Boyd-Graber, K.~Seppi, and L.~Findlater, ``Closing the loop: User-centered design and evaluation of a human-in-the-loop topic modeling system,'' in \emph{Proceedings of the 23rd International Conference on Intelligent User Interfaces}, ser. IUI '18.\hskip 1em plus 0.5em minus 0.4em\relax New York, NY, USA: Association for Computing Machinery, 2018, p. 293–304.

\bibitem{B24}
E.~Horvitz, ``Principles of mixed-initiative user interfaces,'' in \emph{Proceedings of the SIGCHI Conference on Human Factors in Computing Systems}, ser. CHI '99.\hskip 1em plus 0.5em minus 0.4em\relax New York, NY, USA: Association for Computing Machinery, 1999, p. 159–166.

\bibitem{B16}
A.~Endert, P.~Fiaux, and C.~North, ``Semantic interaction for visual text analytics,'' in \emph{Proceedings of the SIGCHI Conference on Human Factors in Computing Systems}, ser. CHI '12.\hskip 1em plus 0.5em minus 0.4em\relax New York, NY, USA: Association for Computing Machinery, 2012, p. 473–482.

\bibitem{B17}
A.~Endert, C.~Han, D.~Maiti, L.~House, S.~Leman, and C.~North, ``Observation-level interaction with statistical models for visual analytics,'' in \emph{2011 IEEE Conference on Visual Analytics Science and Technology (VAST)}, 2011, pp. 121--130.

\bibitem{B47}
E.~Wall, S.~Das, R.~Chawla, B.~Kalidindi, E.~T. Brown, and A.~Endert, ``Podium: Ranking data using mixed-initiative visual analytics,'' \emph{IEEE Transactions on Visualization and Computer Graphics}, vol.~24, no.~1, pp. 288--297, 2018.

\bibitem{B15}
A.~Endert, \emph{Semantic interaction for visual analytics: inferring analytical reasoning for model steering}.\hskip 1em plus 0.5em minus 0.4em\relax Springer Nature, 2022.

\bibitem{B26}
H.~Kim, J.~Choo, H.~Park, and A.~Endert, ``Interaxis: Steering scatterplot axes via observation-level interaction,'' \emph{IEEE Transactions on Visualization and Computer Graphics}, vol.~22, no.~1, pp. 131--140, 2016.

\bibitem{B30}
B.~C. Kwon, H.~Kim, E.~Wall, J.~Choo, H.~Park, and A.~Endert, ``Axisketcher: Interactive nonlinear axis mapping of visualizations through user drawings,'' \emph{IEEE Transactions on Visualization and Computer Graphics}, vol.~23, no.~1, pp. 221--230, 2017.

\bibitem{B38}
T.~Mühlbacher, H.~Piringer, S.~Gratzl, M.~Sedlmair, and M.~Streit, ``Opening the black box: Strategies for increased user involvement in existing algorithm implementations,'' \emph{IEEE Transactions on Visualization and Computer Graphics}, vol.~20, no.~12, pp. 1643--1652, 2014.

\bibitem{vldb_Hasani}
S.~Hasani, S.~Thirumuruganathan, A.~Asudeh, N.~Koudas, and G.~Das, ``Efficient construction of approximate ad-hoc ml models through materialization and reuse,'' \emph{Proceedings of the VLDB Endowment}, vol.~11, pp. 1468--1481, 2018.

\bibitem{NYtimes_and_PubMED}
D.~Newman, ``{Bag of Words},'' UCI Machine Learning Repository, 2008, {DOI}: https://doi.org/10.24432/C5ZG6P.

\bibitem{amazon}
J.~J. McAuley and J.~Leskovec, ``From amateurs to connoisseurs: Modeling the evolution of user expertise through online reviews,'' in \emph{Proceedings of the 22nd International Conference on World Wide Web}, 2013, p. 897–908.

\bibitem{topk-survey}
I.~F. Ilyas, G.~Beskales, and M.~A. Soliman, ``A survey of top-k query processing techniques in relational database systems,'' \emph{ACM Comput. Surv.}, vol.~40, no.~4, 2008.

\bibitem{CuLDA}
X.~Xie, Y.~Liang, X.~Li, and W.~Tan, ``Culda: Solving large-scale lda problems on gpus,'' in \emph{Proceedings of the 28th International Symposium on High-Performance Parallel and Distributed Computing}, 2019, p. 195–205.

\bibitem{LDA*}
L.~Yut, C.~Zhang, Y.~Shao, and B.~Cui, ``Lda*: A robust and large-scale topic modeling system,'' \emph{Proc. VLDB Endowment}, vol.~10, no.~11, p. 1406–1417, 2017.

\bibitem{WarpLDA}
J.~Chen, K.~Li, J.~Zhu, and W.~Chen, ``Warplda: A cache efficient o(1) algorithm for latent dirichlet allocation,'' \emph{Proc. VLDB Endowment}, vol.~9, no.~10, p. 744–755, 2016.

\bibitem{AliasLDA}
A.~Q. Li, A.~Ahmed, S.~Ravi, and A.~J. Smola, ``Reducing the sampling complexity of topic models,'' in \emph{Proc. ACM SIGKDD Int. Conf. Knowl. Discov. Data Mining}, 2014, p. 891–900.

\bibitem{F+LDA}
H.-F. Yu, C.-J. Hsieh, H.~Yun, S.~Vishwanathan, and I.~S. Dhillon, ``A scalable asynchronous distributed algorithm for topic modeling,'' in \emph{Proceedings of the 24th International Conference on World Wide Web}, 2015, p. 1340–1350.

\bibitem{LightLDA}
J.~Yuan, F.~Gao, Q.~Ho, W.~Dai, J.~Wei, X.~Zheng, E.~P. Xing, T.-Y. Liu, and W.-Y. Ma, ``Lightlda: Big topic models on modest computer clusters,'' in \emph{Proceedings of the 24th International Conference on World Wide Web}, 2015, p. 1351–1361.

\bibitem{SVB}
T.~Broderick, N.~Boyd, A.~Wibisono, A.~C. Wilson, and M.~I. Jordan, ``Streaming variational bayes,'' in \emph{Advances in Neural Information Processing Systems}, C.~Burges, L.~Bottou, M.~Welling, Z.~Ghahramani, and K.~Weinberger, Eds., vol.~26, 2013.

\bibitem{SGS}
Y.~Gao, J.~Chen, and J.~Zhu, ``Streaming gibbs sampling for {LDA} model,'' \emph{CoRR}, vol. abs/1601.01142, 2016.

\bibitem{LDA_survey_ACM}
U.~Chauhan and A.~Shah, ``Topic modeling using latent dirichlet allocation: A survey,'' \emph{ACM Comput. Surv.}, vol.~54, no.~7, 2021.

\bibitem{LDA_survey_IJACSA}
R.~Alghamdi and K.~Alfalqi, ``A survey of topic modeling in text mining,'' \emph{International Journal of Advanced Computer Science and Applications}, vol.~6, no.~1, 2015.

\bibitem{SDA_bnp}
T.~Campbell, J.~Straub, J.~W. Fisher~III, and J.~P. How, ``Streaming, distributed variational inference for bayesian nonparametrics,'' in \emph{Advances in Neural Information Processing Systems}, C.~Cortes, N.~Lawrence, D.~Lee, M.~Sugiyama, and R.~Garnett, Eds., vol.~28, 2015.

\bibitem{HDP_nips}
Y.~Teh, M.~Jordan, M.~Beal, and D.~Blei, ``Sharing clusters among related groups: Hierarchical dirichlet processes,'' in \emph{Advances in Neural Information Processing Systems}, L.~Saul, Y.~Weiss, and L.~Bottou, Eds., vol.~17, 2004.

\bibitem{nips10_online}
M.~D. Hoffman, D.~M. Blei, and F.~Bach, ``Online learning for latent dirichlet allocation,'' in \emph{Advances in Neural Information Processing Systems}, vol.~23, 2010.

\bibitem{lda}
D.~M. Blei, A.~Y. Ng, and M.~I. Jordan, ``Latent dirichlet allocation,'' \emph{J. Mach. Learn. Res.}, vol.~3, p. 993–1022, 2003.

\bibitem{jmlr_online_17}
C.~Dupuy and F.~Bach, ``Online but accurate inference for latent variable models with local gibbs sampling,'' \emph{J. Mach. Learn. Res.}, vol.~18, no.~1, p. 4581–4625, 2017.

\bibitem{categoryDS1}
R.~Misra, ``News category dataset,'' \emph{arXiv preprint arXiv:2209.11429}, 2022.

\bibitem{categoryDS2}
R.~Misra and J.~Grover, \emph{Sculpting Data for ML: The first act of Machine Learning}.\hskip 1em plus 0.5em minus 0.4em\relax Independently published, 2021.

\bibitem{realnews}
R.~Zellers, A.~Holtzman, H.~Rashkin, Y.~Bisk, A.~Farhadi, F.~Roesner, and Y.~Choi, ``Defending against neural fake news,'' in \emph{Advances in Neural Information Processing Systems}, 2019.

\bibitem{FSP1}
P.~Ahmed, M.~Hasan, A.~Kashyap, V.~Hristidis, and V.~J. Tsotras, ``Efficient computation of top-k frequent terms over spatio-temporal ranges,'' in \emph{Proc. ACM SIGMOD Int. Conf. Manage. Data}, 2017, p. 1227–1241.

\bibitem{FSP2}
J.~Shang, J.~Liu, M.~Jiang, X.~Ren, C.~R. Voss, and J.~Han, ``Automated phrase mining from massive text corpora,'' \emph{IEEE Transactions on Knowledge and Data Engineering}, vol.~30, no.~10, pp. 1825--1837, 2018.

\bibitem{FSP3}
Z.~He, L.~Wang, C.~Lu, Y.~Jing, K.~Zhang, W.~Han, J.~Li, C.~Liu, and X.~S. Wang, ``Efficiently answering top-k frequent term queries in temporal-categorical range,'' \emph{Information Sciences}, vol. 574, pp. 238--258, 2021.

\bibitem{scikit-learn}
F.~Pedregosa, G.~Varoquaux, A.~Gramfort, V.~Michel, B.~Thirion, O.~Grisel, M.~Blondel, P.~Prettenhofer, R.~Weiss, V.~Dubourg, J.~Vanderplas, A.~Passos, D.~Cournapeau, M.~Brucher, M.~Perrot, and {{\'E}}douard Duchesnay, ``Scikit-learn: Machine learning in python,'' \emph{Journal of Machine Learning Research}, vol.~12, no.~85, pp. 2825--2830, 2011.

\bibitem{sparseLDA}
L.~Yao, D.~Mimno, and A.~McCallum, ``Efficient methods for topic model inference on streaming document collections,'' in \emph{Proc. ACM SIGKDD Int. Conf. Knowl. Discov. Data Mining}, 2009, p. 937–946.

\bibitem{fastLDA}
I.~Porteous, D.~Newman, A.~Ihler, A.~Asuncion, P.~Smyth, and M.~Welling, ``Fast collapsed gibbs sampling for latent dirichlet allocation,'' in \emph{Proc. ACM SIGKDD Int. Conf. Knowl. Discov. Data Mining}, 2008, p. 569–577.

\bibitem{AD-LDA}
D.~Newman, P.~Smyth, M.~Welling, and A.~Asuncion, ``Distributed inference for latent dirichlet allocation,'' in \emph{Advances in Neural Information Processing Systems}, J.~Platt, D.~Koller, Y.~Singer, and S.~Roweis, Eds., vol.~20, 2007.

\bibitem{YahooLDA}
A.~Smola and S.~Narayanamurthy, ``An architecture for parallel topic models,'' \emph{Proc. VLDB Endowment}, vol.~3, no. 1–2, p. 703–710, 2010.

\bibitem{PLDA}
Y.~Wang, H.~Bai, M.~Stanton, W.-Y. Chen, and E.~Y. Chang, ``Plda: Parallel latent dirichlet allocation for large-scale applications,'' in \emph{Algorithmic Applications in Management}, 2009.

\bibitem{ELDA}
S.~Wang, D.~Li, H.~Yu, and H.~Liu, ``Elda: Lda made efficient via algorithm-system codesign submission,'' in \emph{Proceedings of the 25th ACM SIGPLAN Symposium on Principles and Practice of Parallel Programming}, 2020, p. 407–408.

\bibitem{AQP}
M.~N. Garofalakis and P.~B. Gibbon, ``Approximate query processing: Taming the terabytes,'' in \emph{Proc. VLDB Endowment}, 2001, p. 725.

\bibitem{DataCube}
J.~Gray, A.~Bosworth, A.~Lyaman, and H.~Pirahesh, ``Data cube: a relational aggregation operator generalizing group-by, cross-tab, and sub-totals,'' in \emph{Proceedings of the Twelfth International Conference on Data Engineering}, 1996, pp. 152--159.

\bibitem{materialization}
C.~Anagnostopoulos and P.~Triantafillou, ``Efficient scalable accurate regression queries in in-dbms analytics,'' in \emph{2017 IEEE 33rd International Conference on Data Engineering (ICDE)}, 2017, pp. 559--570.

\bibitem{PreCube}
B.-C. Chen, L.~Chen, Y.~Lin, and R.~Ramakrishnan, ``Prediction cubes,'' in \emph{Proc. VLDB Endowment}, 2005, p. 982–993.

\bibitem{coreset1}
O.~Bachem, M.~Lucic, and A.~Krause, ``Scalable k -means clustering via lightweight coresets,'' in \emph{Proc. ACM SIGKDD Int. Conf. Knowl. Discov. Data Mining}, 2018, p. 1119–1127.

\bibitem{coreset2}
D.~Feldman, M.~Schmidt, and C.~Sohler, ``Turning big data into tiny data: Constant-size coresets for k-means, pca and projective clustering,'' in \emph{Proceedings of the Twenty-Fourth Annual ACM-SIAM Symposium on Discrete Algorithms}, 2013, p. 1434–1453.

\bibitem{coreset3}
S.~Har-Peled and S.~Mazumdar, ``On coresets for k-means and k-median clustering,'' in \emph{Proceedings of the Thirty-Sixth Annual ACM Symposium on Theory of Computing}, 2004, p. 291–300.

\bibitem{corset4}
M.~Shindler, A.~Wong, and A.~Meyerson, ``Fast and accurate k-means for large datasets,'' in \emph{Advances in Neural Information Processing Systems}, J.~Shawe-Taylor, R.~Zemel, P.~Bartlett, F.~Pereira, and K.~Weinberger, Eds., vol.~24, 2011.

\bibitem{MLMM1}
H.~Miao, A.~Li, L.~S. Davis, and A.~Deshpande, ``Modelhub: Deep learning lifecycle management,'' in \emph{IEEE 33rd International Conference on Data Engineering (ICDE)}, 2017, pp. 1393--1394.

\bibitem{MLMM2}
{H. Miao, A. Li, L. S. Davis, and A. Deshpande}, ``Towards unified data and lifecycle management for deep learning,'' in \emph{IEEE 33rd International Conference on Data Engineering (ICDE)}, 2017, pp. 571--582.

\bibitem{MLMM3}
M.~Vartak, J.~M. F.~da Trindade, S.~Madden, and M.~Zaharia, ``Mistique: A system to store and query model intermediates for model diagnosis,'' in \emph{Proc. ACM SIGMOD Int. Conf. Manage. Data}, 2018, p. 1285–1300.

\bibitem{MLMM4}
M.~Vartak, H.~Subramanyam, W.-E. Lee, S.~Viswanathan, S.~Husnoo, S.~Madden, and M.~Zaharia, ``Modeldb: A system for machine learning model management,'' in \emph{Proceedings of the Workshop on Human-In-the-Loop Data Analytics}, 2016.

\bibitem{kFSTinRelated}
H.~Ahonen, O.~Heinonen, M.~Klemettinen, and A.~Verkamo, ``Applying data mining techniques for descriptive phrase extraction in digital document collections,'' in \emph{Proceedings IEEE International Forum on Research and Technology Advances in Digital Libraries}, 1998, pp. 2--11.

\bibitem{kFST4}
L.~Chen, S.~Shang, Z.~Zhang, X.~Cao, C.~S. Jensen, and P.~Kalnis, ``Location-aware top-k term publish/subscribe,'' in \emph{IEEE 34th International Conference on Data Engineering (ICDE)}, 2018, pp. 749--760.

\bibitem{topk1}
R.~Fagin, A.~Lotem, and M.~Naor, ``Optimal aggregation algorithms for middleware,'' in \emph{Proceedings of the Twentieth ACM SIGMOD-SIGACT-SIGART Symposium on Principles of Database Systems}, 2001, p. 102–113.

\bibitem{topk2}
R.~Fagin, ``Combining fuzzy information from multiple systems (extended abstract),'' in \emph{Proceedings of the Fifteenth ACM SIGACT-SIGMOD-SIGART Symposium on Principles of Database Systems}, 1996, p. 216–226.

\end{thebibliography}

\vfill

\end{document}